%% file: final.tex
\documentclass[twocolumn]{IEEEtran}
\usepackage{amssymb,amsmath,amsthm,amsfonts,epsfig}

\newcommand{\eqdef}{\mbox{$\buildrel\triangle\over =$}}

\newcommand{\RR}{{\ensuremath{\mathbb R}}}

\newcommand{\NN}{{\ensuremath{\mathbb N}}}

\newcommand{\E}{{\mathsf E}}

\newtheorem{prop}{Proposition}

\newtheorem{rem}{Remark}
\newtheorem{assum}{Assumption}

\renewcommand\theenumi{(\roman{enumi})}
\renewcommand\theenumii{(\alph{enumii})}

\begin{document}
\title{A SURE Approach for Digital Signal/Image Deconvolution Problems}

\author{Jean-Christophe Pesquet, Amel Benazza-Benyahia, and, Caroline Chaux
\thanks{Part of this work was presented 
at the
3rd IEEE International Symposium on Communications, Control and Signal 
Processing \cite{ISCCSP_08}.}
\thanks{This work was supported by the Agence Nationale de la Recherche under grant
ANR-05-MMSA-0014-01.}
\thanks{J.-C. Pesquet and C. Chaux are with the Universit\'e Paris-Est,
Laboratoire d'Informatique Gaspard Monge - CNRS,
77454 Marne la Vall\'ee 
Cedex 2, France. E-mail: 
\texttt{jean-christophe.pesquet@univ-paris-est.fr},
\texttt{caroline.chaux@univ-paris-est.fr}.}
\thanks{A. Benazza-Benyahia is with 
URISA, SUP'COM, 
Cit\'e Technologique des Communications, 2083 Ariana, Tunisia,
E-mail: \texttt{benazza.amel@supcom.rnu.tn}.}}

\maketitle
\begin{abstract}
In this paper, we are interested in the classical problem of restoring data degraded by a convolution
and the addition of a white Gaussian noise.
The originality of the proposed approach is two-fold.  
Firstly, we formulate the restoration problem as a nonlinear estimation problem leading to the minimization of a criterion derived from Stein's unbiased quadratic risk estimate. 
Secondly, the deconvolution procedure is performed using any analysis and synthesis frames that can be overcomplete or not. 
New theoretical results concerning the calculation of the variance
of the Stein's risk estimate are also provided in this work.
Simulations carried out on natural images show the good performance of our method w.r.t. conventional wavelet-based restoration methods.
\end{abstract}


\section{Introduction}
\label{sec:intro}
It is well-known that, in many practical situations, one may consider that there are two main sources of signal/image degradation: a convolution often related to the bandlimited nature of the acquisition system  and a contamination by an additive Gaussian noise which may be due to the electronics of the recording and transmission processes. 
For instance, the limited aperture of satellite cameras, the aberrations inherent to optical systems and mechanical vibrations create a blur effect in  remote sensing images \cite{JAIN_89}.
A data restoration task is usually required to reduce these artifacts before any further processing. Many works have been dedicated to the deconvolution of noisy signals \cite{DEMOMENT_89,KATSAGGELOS_91,COMBETTES_93,BERTERO_98}. Designing suitable deconvolution methods is a challenging  task, as inverse problems of practical interest
are often ill-posed. Indeed, 
the convolution operator is usually non-invertible or it is ill-conditioned
and its inverse is thus very sensitive to noise. To cope with the
ill-posed nature of these problems, deconvolution  methods often operate in a transform domain, the transform being expected to make the problem easier to solve. 
 In pioneering works, deconvolution is dealt with in the frequency domain, 
  as the Fourier transform provides a simple representation of filtering operations \cite{HILLERY_91}. However, the Fourier domain has a main shortcoming: sharp transitions in the signal (edges for images) and other localized features do not have a sparse frequency representation. This has motivated the use of 
the Wavelet Transform (WT) \cite{DONOHO_VAGUELETTE_95,ABRAMOVICH_SILVERMAN_98} and its various extensions. Thanks to the good energy compaction and decorrelation properties of the WT, simple shrinkage operations in the wavelet domain can be successfully applied  to discard noisy coefficients \cite{DONOHO_JOHNSTONE_95}.  
To take advantage of both transform domains, it has been suggested to combine frequency based deconvolution approaches with wavelet-based denoising methods, giving birth to a new class of restoration methods. 
The wavelet-vaguelette
method proposed in \cite{DONOHO_VAGUELETTE_95} is based on an inverse filtering technique.
To avoid the amplification of the resulting colored noise component, a 
shrinkage of the filtered wavelet coefficients is performed.
The wavelet-vaguelette method has been refined in \cite{KALIFA_MALLAT_03}
by adapting the wavelet basis to the frequency response of  the degradation filter.
However, the method is not appropriate for recovering signals degraded by arbitrary convolutive operators. An alternative to the wavelet-vaguelette decomposition is the transform presented by Abramovich and Silverman \cite{ABRAMOVICH_SILVERMAN_98}. 
Similar in the spirit to the wavelet-vaguelette deconvolution, a more competitive hybrid approach called Fourier-Wavelet Regularized Deconvolution (ForWaRD) was developed by Neelamani \textit{et al.}: a two-stage shrinkage procedure successively operates in the Fourier and the WT domains, which is applicable to any invertible or non-invertible degradation kernel \cite{NEE_CHOI_BARANIUK_04}. The optimal balance between the amount of
Fourier and wavelet regularization is derived by optimizing an approximate version of the
mean-squared error metric.  A two-step procedure was also presented by Banham and Katsaggelos which employs
a multiscale Kalman filter \cite{BANHAM_KATSA_96}. 
By following a frequency domain approach,
band-limited Meyer's wavelets have been used to estimate degraded signals
through an elegant wavelet restoration
method called WaveD \cite{PICARD_04,QU_ROUTH_KO_06} which is based on minimax arguments.  
In \cite{EUSIPCO_06}, we have proposed an extension of the WaveD method to the multichannel case.

 Iterative wavelet-based thresholding methods relying
on variational approaches for image restoration have also been investigated by several authors. For instance, a deconvolution method was derived under the expectation-maximization framework  in \cite{FIGUEIREDO_NOWAK_03}. In \cite{STARCK_NGUYEN_03}, the complementarity of the wavelet and the curvelet  transforms has been exploited in a regularized scheme involving the total variation. In
\cite{BECT_BLANCFERAUD_04}, an objective function including the total variation, a wavelet coefficient regularization or a mixed
 regularization has been considered and a related 
 projection-based algorithm was derived to compute the solution.  
More recently, the work in \cite{DEMOL_05}
has been extended by
proposing a flexible convex variational framework for solving
inverse problems in which a priori information (e.g., sparsity or probability distribution) is
available about the representation of the target solution in a frame \cite{CC_PES_COMB_07}. In the same way, a new class of iterative shrinkage/thresholding algorithms was proposed in \cite{BIOUCAS_FIGUEIREDO_07}. Its novelty relies on the fact that  the update equation depends on the two previous iterated values. 
In \cite{VONESCH_UNSER_08}, a fast variational deconvolution algorithm was introduced. It consists of minimizing a quadratic data term subject to a regularization on the $\ell^1$-norm of the coefficients of the solution in a Shannon wavelet basis. 
Recently, in \cite{GUERRERO_PORTILLA_08}, a two-step decoupling scheme 
was presented for image deblurring. It starts from a global linear blur compensation by a generalized Wiener filter. Then, a nonlinear denoising is carried out by computing the Bayes least squares Gaussian scale mixtures estimate. 
Note also that an advanced restoration method was developed 
in \cite{DABOV_FOI_08}, which does not operate in the wavelet domain. 

In the same time, much attention was paid to Stein's principle \cite{STEIN_81}
in order to derive estimates of the Mean-Square Error (MSE) in statistical problems
involving an additive Gaussian noise.
The key advantage of Stein's Unbiased Risk Estimate (SURE) is that it does not require a priori knowledge about the statistics of the unknown data, while yielding  an expression of the MSE only depending on the statistics of the observed data.  Hence, it avoids the difficult problem of the estimation of the hyperparameters of some prior distribution, which classically needs to be addressed in Bayesian approaches.\footnote{This does not mean that SURE approaches are superior to Bayesian approaches, which are quite versatile.} Consequently, a SURE approach can be applied by directly parameterizing the estimator and finding the optimal parameters that minimize the MSE estimate. The first efforts in this direction were performed in the context of denoising applications with the SUREShrink technique \cite{DONOHO_JOHNSTONE_95,KRIM_99} and, the SUREVect estimate \cite{TIP_05} in the case of multichannel images. More recently, in addition to the estimation of the MSE, Luisier \textit{et al.} have proposed a very appealing structure of the denoising function consisting of a linear combination of nonlinear elementary functions (the SURE-Linear Expansion of Threshold or SURE-LET) \cite{BLU_LUISIER_07}. Notice that this idea
was also present in some earlier works \cite{PES_LEPO_97}.
In this way, the optimization of the MSE estimate reduces to solving a set of linear equations.   Several variations of the basic SURE-LET method were investigated: an improvement of the denoising performance has been achieved by accounting for the interscale information \cite{BLU_LUISIER_07_interscale} and  the case of color images has also been addressed \cite{BLU_LUISIER_08}. 
Another advantage of this method is that it remains valid when redundant multiscale representations of the observations are considered, as the minimization of the SURE-LET estimator can be easily carried out in the time/space domain.
A similar approach has also been adopted by Raphan and Simoncelli \cite{RAPHAN_SIMONC_08} for denoising in redundant multiresolution representations. 
Overcomplete representations have also been successfully used for multivariate shrinkage estimators optimized with a SURE approach operating in the transform domain \cite{CHAUX_DUVAL_BENAZZA_PESQUET_08}. An alternative use of Stein's principle was made in \cite{COMBETTES_PESQUET_04} for building convex constraints in image denoising problems.

In \cite{ELDAR_08},
Eldar generalized Stein's principle to derive  an MSE estimate  when the noise has an exponential distribution  (see also \cite{AVERKAMP_HOUDRE_06}). In addition, she investigated the problem of the nonlinear estimation of deterministic parameters
from a linear observation model in the presence of additive noise. 
In the context of  deconvolution, the derived SURE was employed to evaluate the MSE performance of solutions to regularized objective functions.
Another work in this research direction is \cite{RAMANI_BLU_UNSER_08},
where the risk estimate is minimized by a Monte Carlo technique for
denoising applications. A very recent work \cite{VONESH_RAMANI_UNSER_08} also proposes a recursive estimation
of the risk when a thresholded Landweber algorithm is employed to restore data.

In this paper, we adopt a viewpoint similar to that in \cite{ELDAR_08,VONESH_RAMANI_UNSER_08}
in the sense that, by using Stein's principle, we obtain an estimate of the 
MSE for a given class of estimators operating 
in deconvolution problems. The main contribution of our work
is the derivation of the variance of the proposed quadratic risk estimate.
These results allow us to propose a novel SURE-LET approach for data restoration which can exploit any discrete frame representation.  

 The paper is organized as follows. In Section \ref{sec:statement_of_problem}, the required background is presented and
some notations are introduced. The generic form of the estimator we consider
for restoration purposes is presented in Section~\ref{se:nonlinest}.
In Section \ref{s:steinid}, we provide extensions of Stein's identity which
will be useful throughout the paper.
In Section~\ref{sec:inv} we show how Stein's principle can be employed in a restoration framework when the degradation system is invertible.  The case of a non-invertible system is addressed in Section \ref{se:noninv}.
The expression of the variance of the empirical estimate of the quadratic risk is then derived in Section \ref{sec:varemp}.
In Section \ref{se:flexiblity}, two scenarii are discussed where 
the determination of the parameters minimizing the risk estimate
takes a simplified form. The structure of the proposed SURE-LET deconvolution method
is subsequently described in Section \ref{se:param} and
examples of its application to wavelet-based image restoration
are shown in Section \ref{sec:simuls}. Some concluding remarks are given
in Section~\ref{se:conclu}.

The notations used in the paper are summarized in Table \ref{tab:notations}.

\begin{table*}[h!tb]
\caption{Notations. \label{tab:notations}}
\begin{small}
\begin{center}
\begin{tabular}{|l|l|}
\hline
Variable&Definition \\
\hline
$\mathbb{D}$&set of spatial (or frequency) indices\\
\hline
$s$&original signal\\
\hline
$h$&impulse response of the degradation filter\\
\hline
$n$&additive white Gaussian noise\\
\hline
$r$&observed signal\\
\hline
$S$&Fourier transform of $s$\\
\hline
$H$&Fourier transform of $h$\\
\hline
$R$&Fourier transform of $r$\\
\hline
$\mathcal{E}(\cdot)$ & mean square value of the signal in argument\\
\hline
$\E[\cdot]$ & mathematical expectation\\
\hline
$(\varphi_\ell)_{1 \le \ell \le L}$ & family of  analysis vectors\\
\hline
$(\widetilde{\varphi}_\ell)_{1\le \ell\le L}$& family of  synthesis vectors\\
\hline
$(s_\ell)_{1\le \ell\le L}$&coefficients of the decomposition of $s$ onto $(\varphi_\ell)_{1 \le \ell \le L}$\\
\hline
${\Phi}_{\ell}$&Fourier transform of $\varphi_\ell$\\
\hline
$\widetilde{\Phi}_{\ell}$&Fourier transform of $\widetilde{\varphi}_\ell$\\
\hline
$\Theta_{\ell}$& estimating function applied to $s_\ell$\\
\hline
$\widehat{s}$&estimate of $s$\\
\hline
$\mathbb{P}$&set of frequency indices for which $H$ is considered equal to 0\\
\hline
$\mathbb{Q}$&set of frequency indices for which $H$ takes significant values\\
\hline
$\chi$ & threshold value in the frequency domain\\
\hline
$\underline{s}$&projection of $s$ onto the subspace whose\\
& Fourier coefficients vanish on $\mathbb{P}$\\
\hline 
$\underline{\widetilde{r}}$&inverse Fourier transform of the projection of $\frac{R}{H}$ \\
&onto the subspace whose Fourier coefficients vanish on $\mathbb{P}$\\
\hline
$\mathbb{K}_m$ & index subset for the $m$-th subband\\
\hline
$\lambda$ & constant used in the Wiener-like filter\\
\hline
\end{tabular}\end{center}
\end{small}
\end{table*}

\section{Problem statement}

\subsection{Background}\label{sec:statement_of_problem}
We consider an unknown real-valued field whose value at location ${\mathbf{x}} \in\mathbb{D}$ is $s(\mathbf{x})$ where 
$\mathbb{D} = \{0,\ldots,D_1-1\}\times \cdots \times \{0,\ldots,D_d-1\}$
with $(D_1,\ldots,D_d) \in (\mathbb{N}^*)^d$ where $\NN^*$ denotes the set of positive integers.
Here, $s$ is a $d$-dimensional digital random field of finite size $D = D_1\ldots D_d$ with finite variance.
Of pratical interest are the cases when $d = 1$ (temporal signals), $d=2$ (images),
$d= 3$ (volumetric data or video sequences) and $d=4$ (3D$+t$ data).

The field is degraded by the acquisition
 system with (deterministic) impulse response $h$, and
it is also corrupted by an additive noise
$n$, which is assumed to be independent of the
random process $s$. The noise $n$
corresponds to a random field which is assumed to be Gaussian
with zero-mean and covariance field:
$
\forall (\mathbf{x},\mathbf{y}) \in \mathbb{D}^2$,
$\E[n(\mathbf{x}) n(\mathbf{y})] = \gamma
\delta_{\mathbf{x}-\mathbf{y}}
$,
where $(\delta_{\mathbf{x}})_{\mathbf{x}\in \mathbb{Z}^d}$ is the Kronecker sequence and $\gamma > 0$.
In other words, the noise is white.

Thus, the observation model can be expressed as follows:
 \begin{equation}
\forall \mathbf{x} \in \mathbb{D},\quad
r({\mathbf{x}}) =
(\widetilde{h}*s)({\mathbf{x}})+n({\mathbf{x}})
= \sum_{\mathbf{y} \in \mathbb{D}} \widetilde{h}(\mathbf{x}-\mathbf{y})
s(\mathbf{y})
+ n({\mathbf{x}})
\label{eq:blur_noise}
 \end{equation}
where $(\widetilde{h}(\mathbf{x}))_{\mathbf{x} \in \mathbb{Z}^d}$ is the periodic
extension of  $(h(\mathbf{x}))_{\mathbf{x} \in \mathbb{D}}$.
It must be pointed out that \eqref{eq:blur_noise}
corresponds to a periodic  approximation of the discrete convolution 
(this problem can be alleviated by making use of zero-padding techniques \cite{BRIGHAM_74,GONZALEZ_WOODS_92}).
 
A restoration method aims at estimating $s$
 based on the observed data $r$.
In this paper, a supervised approach is adopted by assuming that both the degradation kernel $h$
and the noise variance $\gamma$ are known.  

\subsection{Considered nonlinear estimator}
\label{se:nonlinest}
The proposed estimation procedure consists of first transforming the observed data to some other domain (through some \emph{analysis} vectors), performing a non-linear operation on the so-obtained coefficients (based on an \emph{estimating function}) with parameters that must be estimated, and finally reconstructing the estimated signal (through some \emph{synthesis} vectors).

More precisely, the discrete Fourier coefficients $\big(R({\mathbf{p}})\big)_{\mathbf{p}\in \mathbb{D}}$ of $r$ are given by:
\begin{equation}
\forall {\mathbf{p}}\in \mathbb{D},\qquad
R({\mathbf{p}})\,\eqdef\,\sum_{\mathbf{x}\in \mathbb{D}}
r({\mathbf{x}})\exp(-2\pi \imath
{\mathbf{x}}^\top{\boldsymbol D}^{-1}{\mathbf{p}}) \label{eq:FT}
\end{equation}
where ${\boldsymbol D} = \mathrm{Diag}(D_1,\ldots,D_d)$.
In the frequency domain, (\ref{eq:blur_noise}) becomes:
\begin{equation}
R({\mathbf{p}})=
U({\mathbf{p}})+N({\mathbf{p}}),
\quad\text{where}\quad
U({\mathbf{p}})\,\eqdef\,H({\mathbf{p}})S({\mathbf{p}})
\label{eq:convperfreq}
\end{equation}
and,  the coefficients $S({\mathbf{p}})$ and
$N({\mathbf{p}})$ are obtained by expressions similar to (\ref{eq:FT}).

Let $(\varphi_\ell)_{1 \le \ell \le L}$ be a family of $L \in \NN^*$ analysis vectors
of $\mathbb{R}^{D_1\times \cdots \times D_d}$. Thus, every signal $r$ of $\mathbb{R}^{D_1\times \cdots \times D_d}$ can be decomposed as:
\begin{equation}
\forall \ell \in \{1,\ldots,L\},\quad
r_{\ell}
= \langle r, \varphi_{\ell } \rangle =
\sum_{\mathbf{x}\in \mathbb{D}}
r({\mathbf{x}})\varphi_{\ell }({\mathbf{x}}),
\end{equation}
the operator $\langle \cdot, \cdot \rangle$ designating the Euclidean inner product of
 $\mathbb{R}^{D_1\times \cdots \times D_d}$. 
According to 
Plancherel's formula, the coefficients of the decomposition
of $r$ onto this family are given by
\begin{equation}
\forall \ell \in \{1,\ldots,L\},\quad
r_{\ell}
=\frac{1}{D}\sum_{{\mathbf{p}\in \mathbb{D}}}R({\mathbf{p}}
)\big({\Phi}_{\ell}({\mathbf{p}})\big)^*,
\label{eq:Plancherel}
\end{equation}
where ${\Phi}_{\ell}({\mathbf{p}})$ is a discrete
Fourier coefficient of $\varphi_{\ell }$ and $(\cdot)^*$ denotes the complex
conjugation.
Let us now define, for every $\ell \in \{1,\ldots,L\}$, an estimating function $\Theta_{\ell}\colon \RR \to \RR$ (the choice of this function will be discussed in Section \ref{se:choiceest}),
so that
\begin{equation}
\widehat{s}_{\ell} = \Theta_{\ell}\big(r_{\ell}\big).
\label{eq:NLestpart}
\end{equation}
We will use as an estimate of $s(\mathbf{x})$,
\begin{equation}
\widehat{s}(\mathbf{x}) = \sum_{\ell=1}^L \widehat{s}_{\ell}\,\widetilde{\varphi}_{\ell}(\mathbf{x})
\label{eq:estsynth}
\end{equation}
where $(\widetilde{\varphi}_\ell)_{1\le \ell\le L}$ is a family of synthesis vectors
of $\mathbb{R}^{D_1\times \cdots \times D_d}$.
Equivalently, the estimate of $S$ is given by:
\begin{equation}
\widehat{S}(\mathbf{p}) = \sum_{\ell=1}^L \widehat{s}_{\ell}\,\widetilde{\Phi}_{\ell}(\mathbf{p})
\label{eq:estsynthf}
\end{equation}
where ${\widetilde{\Phi}}_{\ell}({\mathbf{p}})$ is a discrete
Fourier coefficient of $\widetilde{\varphi}_{\ell }$.
It must be pointed out that our formulation is quite general. Different analysis/synthesis families can be used. 
 These families may be overcomplete (which implies that $L > D$)
or not.

\section{Stein-like identities}\label{s:steinid}
Stein's principle will play a central role in the evaluation of the 
mean square estimation error of the proposed estimator.
We first recall the standard form of Stein's principle:
\begin{prop}{\rm \cite{STEIN_81}}\label{p:stein0}
Let $\Theta\colon \RR \to \RR$ be a continuous, almost everywhere differentiable function. Let $\eta$ be a real-valued zero-mean Gaussian random variable
with variance $\sigma^2$
and $\upsilon$ be a real-valued random variable which is independent of
$\eta$. Let $\rho = \upsilon+\eta$ and
assume that
\begin{itemize}
\item $\forall \tau \in \RR$, 
$\lim_{|\zeta| \to \infty} \Theta(\tau+\zeta) 
\exp\big(-\frac{\zeta^2}{2\sigma^2}\big) = 0$,
\item $\E[(\Theta(\rho))^2] < \infty$ and $\E[|\Theta'(\rho)|] < \infty$
where $\Theta'$ is the derivative of $\Theta$. 
\end{itemize}
Then,
\begin{equation}
\E[ \Theta(\rho)\eta]
= \sigma^2\E[ \Theta'(\rho)].
\label{eq:steinf0}
\end{equation}
\end{prop}
We now derive extended forms of the above formula (see Appendix \ref{a:steinf}) which will be useful
in the remainder of this paper:
\begin{prop} \label{p:steinf}
Let $\Theta_i\colon \RR \to \RR$ with $i\in \{1,2\}$ be continuous, almost everywhere differentiable functions.
Let $(\eta_1,\eta_2,\widetilde{\eta}_1,\widetilde{\eta}_2)$ be a real-valued 
zero-mean Gaussian vector
and $(\upsilon_1,\upsilon_2)$ be a real-valued random vector which is independent of
$(\eta_1,\eta_2,\widetilde{\eta}_1,\widetilde{\eta}_2)$. Let $\rho_i = \upsilon_i+\eta_i$ where $i\in \{1,2\}$ and
assume that
\begin{enumerate}
\item\label{as:psteinf1} $\forall \alpha \in \RR^*$, $\forall \tau \in \RR$, 
$\lim_{|\zeta| \to \infty} \Theta_i(\tau+\zeta)\zeta^2
\exp\big(-\frac{\zeta^2}{2\alpha^2}\big) = 0$,
\item\label{as:psteinf2} $\E[|\Theta_i(\rho_i)|^3] < \infty$,
\item\label{as:psteinf3} $\E[|\Theta_i'(\rho_i)|^3] < \infty$
where $\Theta'_i$ is the derivative of $\Theta_i$. 
\end{enumerate}
Then,
\begin{align}
\E[ \Theta_1(\rho_1)\widetilde{\eta}_1]
= &\,\E[ \Theta_1'(\rho_1)] \E[\eta_1\widetilde{\eta}_1]\label{eq:steinf1}\\
\E[ \Theta_1(\rho_1)\widetilde{\eta}_1\widetilde{\eta}_2]
=&\,\E[\Theta_1'(\rho_1)\widetilde{\eta}_2]\E[\eta_1\widetilde{\eta}_1]\nonumber\\
&+\E[\Theta_1(\rho_1)]\E[\widetilde{\eta}_1\widetilde{\eta}_2]\label{eq:steinf3}\\
\E[ \Theta_1(\rho_1)\widetilde{\eta}_1\widetilde{\eta}_2^2]=
&\,\E[ \Theta_1'(\rho_1)\widetilde{\eta}_2^2]\E[\eta_1\widetilde{\eta}_1]
+2\E[\Theta_1'(\rho_1)]\nonumber\\
&\times \E[\widetilde{\eta}_1\widetilde{\eta}_2]
\E[\widetilde{\eta}_2\eta_1]. \label{eq:steinf4}\\
\E[\Theta_1(\rho_1)\Theta_2(\rho_2)\widetilde{\eta}_1\widetilde{\eta}_2] =
&\,\E[\Theta_1(\rho_1)\Theta_2(\rho_2)]\E[\widetilde{\eta}_1\widetilde{\eta}_2]\nonumber\\
&+\E[\Theta_1'(\rho_1)\Theta_2(\rho_2)\widetilde{\eta}_2]\E[\eta_1\widetilde{\eta}_1]\nonumber\\
&+\E[\Theta_1(\rho_1)\Theta'_2(\rho_2)\widetilde{\eta}_1]
\E[\eta_2\widetilde{\eta}_2]\nonumber\\
&+\E[\Theta_1'(\rho_1)\Theta'_2(\rho_2)] (\E[\eta_1\widetilde{\eta}_2]
\E[\eta_2\widetilde{\eta}_1]\nonumber\\
&-\E[\eta_1\widetilde{\eta}_1]
\E[\eta_2\widetilde{\eta}_2]).
\label{eq:steinf5}
\end{align}
\end{prop}
Note that
Proposition \ref{p:steinf} obviously is applicable when $(\upsilon_1,\upsilon_2)$ is deterministic. 

\section{Use of Stein's principle}\label{sec:appli_stein}
\subsection{Case of invertible degradation systems}\label{sec:inv}
In this section, we come back to the deconvolution problem
and develop an unbiased estimate of
the quadratic risk:
\begin{equation}
\mathcal{E}(\widehat{s}-s)  = \frac{1}{D}
\sum_{\mathbf{x}\in \mathbb{D}} \big(s(\mathbf{x})-\widehat{s}(\mathbf{x})\big)^2 \label{eq:defE}
\end{equation}
which will be useful to optimize a parametric form of the estimator 
from the observed data.
For this purpose, the following assumption is made:
\begin{assum}
\begin{enumerate}\ \label{as:stein}
\item \label{as:invfilt} The degradation filter is such that,
for every $\mathbf{p}\in \mathbb{D}$, $H(\mathbf{p}) \neq 0$.
\item \label{as:steinhyp} For every $\ell$ in $\{1,\ldots,L\}$, $\Theta_{\ell}$ is a continuous, almost everywhere differentiable function such that
\begin{enumerate}
\item $\forall \alpha \in \RR^*$,$\forall \tau \in \RR$,\\ 
$\displaystyle \lim_{|\zeta| \to \infty} \Theta_\ell(\tau+\zeta)\zeta^2 
\exp(-\frac{\zeta^2}{2\alpha^2}) = 0$,
\item $\E[|\Theta_{\ell}(r_\ell)|^3] < \infty$ 
and $\E[|\Theta'_{\ell}(r_\ell)|^3] < \infty$
where $\Theta_\ell'$ is the derivative of $\Theta_\ell$.
\end{enumerate}
\end{enumerate}
\end{assum}
Under this assumption, the degradation model can be re-expressed as
$
s(\mathbf{x}) = \widetilde{r}(\mathbf{x}) -
 \widetilde{n}(\mathbf{x})
$
where $\widetilde{r}$ and $\widetilde{n}$ are the fields whose discrete Fourier coefficients  are
\begin{equation}
\widetilde{R}(\mathbf{p}) = 
\frac{R(\mathbf{p})}{H(\mathbf{p})},\quad
\widetilde{N}(\mathbf{p}) = 
\frac{N(\mathbf{p})}{H(\mathbf{p})}.\label{eq:CFrt}
\end{equation}
Thus, since the noise has been assumed spatially white,
it is easy to show that
\begin{align}
\forall (\mathbf{p},\mathbf{p}') \in \mathbb{D}^2,\qquad
\E\big[\widetilde{N}(\mathbf{p})\big(N(\mathbf{p}')\big)^*\big]
=&\frac{\gamma D}{H(\mathbf{p})} \delta_{\mathbf{p}-\mathbf{p}'}
\label{eq:noisestatFI}\\
\E\big[\widetilde{N}(\mathbf{p})\big(\widetilde{N}(\mathbf{p'})\big)^*\big]
=& \frac{\gamma D}{|H(\mathbf{p})|^2} \delta_{\mathbf{p}-\mathbf{p}'}
\label{eq:noisestatF}
\end{align}
and $\E\big[\widetilde{N}(\mathbf{p})\big(S(\mathbf{p'})\big)^*\big] = 0$.
The latter relation shows that $\widetilde{n}$ and $s$ are uncorrelated fields.

We are now able to state the following result (see Appendix \ref{a:steininv}):
\begin{prop}\label{p:steininv}
The mean square error on each frequency component is such that,
for every $\mathbf{p} \in \mathbb{D}$,
\begin{multline}
\E[|\widehat{S}(\mathbf{p})-S(\mathbf{p})|^2]
= \E[|\widehat{S}(\mathbf{p})-\widetilde{R}(\mathbf{p})|^2]
-\frac{\gamma D}{| H(\mathbf{p}) |^2}\\
+2\gamma\sum_{\ell=1}^L \E[\Theta_\ell'(r_\ell)] 
\,\mathrm{Re}\Big\{\frac{\Phi_\ell(\mathbf{p}) \big(\widetilde{\Phi}_\ell(\mathbf{p})\big)^*}{H(\mathbf{p})}\Big\}
\label{eq:steindecriskDp}
\end{multline}
and, the global mean square estimation error can be expressed as
\begin{align}
&\E[\mathcal{E}(\widehat{s}-s)]  =
\E[\mathcal{E}(\widehat{s}-\widetilde{r})]
+ \Delta \label{eq:steindecrisk}\\
&\Delta = \frac{\gamma}{D}\Big(2\sum_{\ell=1}^L \E[\Theta_\ell'(r_\ell)]
\overline{\gamma}_\ell
-\sum_{\mathbf{p}\in \mathbb{D}} | H(\mathbf{p}) |^{-2}\Big)
\label{eq:steindecriskD}
\end{align}
where $(\overline{\gamma}_\ell)_{1 \le \ell \le L}$ is the real-valued cross-correlation sequence defined by: for all $\ell \in \{1,\ldots,L\}$,
\begin{equation}
\overline{\gamma}_\ell
 = \frac{1}{D} \sum_{\mathbf{p}\in\mathbb{D}}
\frac{\Phi_\ell(\mathbf{p}) \big(\widetilde{\Phi}_\ell(\mathbf{p})\big)^*}{H(\mathbf{p})}.
\label{eq:devintertermf}
\end{equation}
\end{prop}

\subsection{Case of non-invertible degradation systems}
\label{se:noninv}
Assumption \ref{as:stein}\ref{as:invfilt} expresses the fact that the degradation filter
is invertible. Let us now examine how this assumption can be relaxed.

We denote by $\mathbb{P}$ the set of indices for which the frequency response $H$ vanishes:
\begin{equation}
\mathbb{P} = \{\mathbf{p}\in \mathbb{D}\;\mid\;
H(\mathbf{p}) = 0\}.
\label{eq:P0}
\end{equation}
It is then clear that the components of $S(\mathbf{p})$ with $\mathbf{p}\in\mathbb{P}$, are unobservable.
The observable part of the signal $s$ thus corresponds to the projection
$\underline{s}=\Pi\big(s\big)$ of $s$ onto the subspace of $\RR^{D_1\times \cdots \times D_d}$
of the fields whose discrete Fourier coefficients vanish
on $\mathbb{P}$. In the Fourier domain, the projector $\Pi$ is therefore defined by
\begin{equation}
\forall \mathbf{p} \in \mathbb{D},\qquad
\underline{S}(\mathbf{p})
= \begin{cases}
S(\mathbf{p}) & \mbox{if $\mathbf{p} \not\in \mathbb{P}$}\\
0 & \mbox{if $\mathbf{p} \in \mathbb{P}$.}
\end{cases}
\label{eq:defprojb}
\end{equation}
In this context, it is judicious to restrict the summation in \eqref{eq:Plancherel} to $\mathbb{Q} = \mathbb{D}\setminus \mathbb{P}$
so as to limit the influence of the noise present in the unobservable
part of $s$. This leads to the following modified expression of
the coefficients $r_\ell$:
\begin{equation}
r_\ell = 
\frac{1}{D}\sum_{\mathbf{p} \in \mathbb{Q}}
R(\mathbf{p}) \big(\Phi_\ell(\mathbf{p})\big)^*=\langle \underline{r},\varphi_\ell\rangle
\label{eq:modrl}
\end{equation}
where $\underline{r} = \Pi\big(r\big)$. The second step in the estimation procedure (Eq. \eqref{eq:NLestpart}) is kept unchanged. For the last step, we 
impose the following structure to the estimator:
\begin{align}
\widehat{s}(\mathbf{x}) &= 
\Pi\Big(\sum_{\ell=1}^L \widehat{s}_\ell\;\widetilde{\varphi}_\ell(\mathbf{x})\Big)
=\sum_{\ell=1}^L \widehat{s}_\ell\;\underline{\widetilde{\varphi}}_\ell(\mathbf{x})
\label{eq:estsu}
\end{align}
where $\underline{\widetilde{\varphi}}_\ell = \Pi(\widetilde{\varphi}_\ell)$.
We will also replace Assumption \ref{as:stein}\ref{as:invfilt} by the following less restrictive one:
\begin{assum}\label{as:Q}
The set $\mathbb{Q}$ is nonempty.
\end{assum}

Under this condition and Assumption \ref{as:stein}\ref{as:steinhyp}, an extended form of Proposition \ref{p:steininv}
is the following:
\begin{prop} \label{prop:ninv}
The mean square error on each frequency component is given,
for every $\mathbf{p} \in \mathbb{Q}$, by \eqref{eq:steindecriskDp}.
The global mean square estimation error can be expressed as
\begin{align}
&\E[\mathcal{E}(\widehat{s}-s)]  = \E[\mathcal{E}(s-\underline{s})]+ \E[\mathcal{E}(\widehat{s}-\underline{\widetilde{r}})]+\Delta
\label{eq:steindecriskiv}\\
&\Delta = \frac{\gamma}{D} \Big(2\sum_{\ell=1}^L
\E[\Theta_\ell'(r_\ell)] \overline{\gamma}_\ell-
\sum_{\mathbf{p}\in \mathbb{Q}} | H(\mathbf{p}) |^{-2}\Big).
\label{eq:steindecriskDiv}
\end{align}
Hereabove, $\underline{\widetilde{r}}$ denotes the 2D field with 
Fourier coefficients
\begin{equation}
\underline{\widetilde{R}}(\mathbf{p})
= \begin{cases}
\displaystyle{\frac{R(\mathbf{p})}{H(\mathbf{p})}} & \mbox{if $\mathbf{p} \in \mathbb{Q}$}\\
0 & \mbox{otherwise},
\end{cases}
\label{eq:defRt}
\end{equation}
and, the real-valued cross-correlation sequence
$(\overline{\gamma}_\ell)_{1 \le \ell\le L}$ becomes:
\begin{equation}
\overline{\gamma}_\ell
 = \frac{1}{D} \sum_{\mathbf{p} \in \mathbb{Q}}
\frac{\Phi_\ell(\mathbf{p}) \big(\widetilde{\Phi}_\ell(\mathbf{p}))^*}{H(\mathbf{p})}.
\label{eq:devintertermfni}
\end{equation}

\end{prop}
\begin{proof}
The proof that  \eqref{eq:steindecriskDp} holds for every $\mathbf{p}\in \mathbb{Q}$ is identical to that in Proposition \ref{p:steininv}.
The global MSE can be decomposed as the sum of the errors on its unobservable and observable parts, respectively. Using the orthogonality property for the projection operator $\Pi$, the corresponding quadratic risk is given by:
$
\mathcal{E}(\widehat{s}-s) = \mathcal{E}(s-\underline{s})
+ \mathcal{E}(\widehat{s}-\underline{s}).
$
It remains now to express the mean square estimation error 
 $\E[\mathcal{E}(\widehat{s}-\underline{s})]$
on the observable part. This is done quite similarly to the end of the proof of Proposition \ref{p:steininv}.
\end{proof}
\begin{rem}\ \label{re:iv}
\begin{enumerate}
\item Assume that the functions $s$ and $h$ share the same frequency band
in the sense that, for all 
$\mathbf{p} \not \in \mathbb{Q}$, $S(\mathbf{p}) = 0$. Let us also assume that, 
for every $\mathbf{p} \in \mathbb{Q}$, $H(\mathbf{p}) = 1$. This typically corresponds to a denoising problems for a signal with frequency band
$\mathbb{Q}$. Then, since $\underline{s} = s$ and $\underline{\widetilde{r}} =
\underline{r}$, \eqref{eq:steindecriskiv} becomes
\begin{multline}
\E[\mathcal{E}(\widehat{s}-s)]=\E[\mathcal{E}(\widehat{s}-\underline{r})]\\
+\frac{\gamma}{D} \Big(2 \sum_{\ell=1}^L \E[\Theta_\ell'(r_\ell)] 
\langle \varphi_\ell,\underline{\widetilde{\varphi}}_\ell\rangle
-\mathrm{card}(\mathbb{Q})\Big),
\end{multline}
where $\mathrm{card}(\mathbb{Q})$ denotes the cardinality of
$\mathbb{Q}$.
In the case when $d=2$ (images) and 
$\mathbb{Q}=\mathbb{D}$, the resulting expression is identical to the one which has been derived in \cite{BLU_LUISIER_07} for denoising problems.
\item Proposition \ref{prop:ninv} remains valid for more general choices of
the set $\mathbb{P}$ 
than \eqref{eq:P0}. In particular, \eqref{eq:steindecriskiv} 
and \eqref{eq:steindecriskDiv} are unchanged if
\begin{equation}
\mathbb{P} =  \{\mathbf{p}\in \mathbb{D}\;\mid\; |H(\mathbf{p})| \le \chi\}
\label{eq:defPchi}
\end{equation}
where $\chi \ge 0$, provided that the complementary set $\mathbb{Q}$
satisfies Assumption \ref{as:Q}.
\item It is possible to give an alternative proof of 
\eqref{eq:steindecriskiv}-\eqref{eq:steindecriskDiv} 
by applying Proposition~1 in \cite{ELDAR_08}. 
\end{enumerate}
\end{rem}

\section{Empirical estimation of the risk}
\label{sec:varemp}
Under the assumptions of the previous section, we are now interested
in the estimation of the ``observable'' part of the risk in  \eqref{eq:steindecriskiv}, that is $\mathcal{E}_o = \mathcal{E}(\widehat{s}-\underline{s})$, from the observed field $r$. As shown by Proposition \ref{prop:ninv}, an unbiased estimator of $\mathcal{E}_o$ is 
\begin{equation}
\widehat{\mathcal{E}}_o =  \mathcal{E}(\widehat{s}-\underline{\widetilde{r}})+\widehat{\Delta}
\label{eq:Eo}
\end{equation}
where
\begin{equation}
\widehat{\Delta}=\frac{\gamma}{D}\Big(2\sum_{\ell=1}^L \Theta_\ell'(r_\ell) \overline{\gamma}_\ell-\sum_{\mathbf{p}\in \mathbb{Q}} | H(\mathbf{p}) |^{-2}\Big).
\label{eq:estsDelta}
\end{equation}
We will study in more detail the statistical behaviour of this estimator
by considering the difference:
\begin{equation}
\mathcal{E}_o-\widehat{\mathcal{E}}_o
= \frac{2}{D} \sum_{\mathbf{x} \in \mathbb{D}} \big(\widehat{s}(\mathbf{x})-\underline{s}(\mathbf{x})\big)\,\underline{\widetilde{n}}(\mathbf{x})-
\mathcal{E}(\underline{\widetilde{n}})-\widehat{\Delta}.
\label{eq:difE0}
\end{equation}
More precisely, by making use of Proposition \ref{p:steinf},
the variance of this term can be derived (see Appendix~\ref{a:varstein}).
\begin{prop}\label{p:varstein}
The variance of the estimate of the observable part of the quadratic risk is given by
\begin{multline}
\mathsf{Var}[\mathcal{E}_o-\widehat{\mathcal{E}}_o]
= \frac{4\gamma}{D}\E[\mathcal{E}(\widehat{s}_H-\widetilde{r}_H)]\\
\quad+\frac{4\gamma^2}{D^2} \sum_{\ell=1}^L
\sum_{i=1}^L\E[\Theta_\ell'(r_\ell)\Theta_i'(r_i)]
\overline{\gamma}_{\ell,i} \overline{\gamma}_{i,\ell}-
\frac{2\gamma^2}{D^2} \sum_{\mathbf{p}\in \mathbb{Q}}\frac{1}{|H(\mathbf{p})|^4}
\label{eq:varsteinf} 
\end{multline}
where $\widetilde{r}_H$ is the field with discrete Fourier coefficients
given by
\begin{equation}
\widetilde{R}_H(\mathbf{p}) = 
\begin{cases}
\displaystyle\frac{\widetilde{R}(\mathbf{p})}{H(\mathbf{p})} & \mbox{if $\mathbf{p} \in
\mathbb{Q}$}\\
0 & \mbox{otherwise,}
\end{cases}
\label{eq:indexH}
\end{equation}
$\widehat{s}_H$ is similarly defined from $\widehat{s}$ and,
\begin{equation}
\forall (\ell,i) \in \{1,\ldots,L\}^2,\qquad
\overline{\gamma}_{\ell,i}
= \frac{1}{D}
\sum_{\mathbf{p}\in \mathbb{Q}}
\frac{\Phi_\ell(\mathbf{p})\big(\widetilde{\Phi}_i(\mathbf{p})\big)^*}{H(\mathbf{p})} .
\label{eq:nf5iii}
\end{equation}
\end{prop}

\begin{rem}
\begin{enumerate}\ 
\item Eq. \eqref{eq:varsteinf} suggests that caution should be taken
in relying on the unbiased risk estimate when $|H(\mathbf{p})|$
takes small values. Indeed, the terms in the expression of the variance
involve divisions by $H(\mathbf{p})$ and may therefore become of high
magnitude, in this case.
\item An alternative statement of Proposition \ref{p:varstein} is to say that
\begin{align*}
\frac{4\gamma}{D}\mathcal{E}(\widehat{s}_H-\widetilde{r}_H)
&+\frac{4\gamma^2}{D^2} \sum_{\ell=1}^L
\sum_{i=1}^L \Theta_\ell'(r_\ell)\Theta_i'(r_i)
\overline{\gamma}_{\ell,i} \overline{\gamma}_{i,\ell}\\
&-\frac{2\gamma^2}{D^2} \sum_{\mathbf{p}\in \mathbb{Q}}|H(\mathbf{p})|^{-4}
\end{align*} 
is an unbiased estimate of $\mathsf{Var}[\mathcal{E}_o-\widehat{\mathcal{E}}_o]$.
\end{enumerate}
\end{rem}
\section{Case study}
\label{se:flexiblity}
It is important to emphasize that the proposed restoration framework presents various degrees of freedom. Firstly, it is possible to choose redundant or non redundant analysis/synthesis families. In Section \ref{subse:ortho_synthesis}, we will show that in the case of orthonormal synthesis families, the estimator design can be split into several simpler optimization procedures. Secondly, any structure of the estimator  can be virtually considered. Of particular interest are restoration methods involving
 Linear Expansion of Threshold (LET) functions, which are investigated in  Section \ref{subse:LET}. As already mentioned, the latter estimators have been successfully used in denoising problems \cite{BLU_LUISIER_07}.

\subsection{Use of orthonormal synthesis families}
\label{subse:ortho_synthesis}
We now examine the case when
$(\underline{\widetilde{\varphi}}_\ell)_{1\le \ell \le L}$
is an orthonormal basis of $\Pi(\RR^{D_1\times\cdots\times  D_d})$ (thus,
$L = \mathrm{card}(\mathbb{Q})$). This arises, in particular,
when 
$(\widetilde{\varphi}_\ell)_{1\le \ell \le L}$ is an orthonormal basis of $\RR^{D_1\times\cdots\times  D_d}$
and the degradation system is invertible ($\mathbb{Q}= \mathbb{D}$).
Then, due to the orthogonality of the functions
$(\underline{\widetilde{\varphi}}_\ell)_{1 \le \ell \le L}$, the unbiased estimate of the risk in \eqref{eq:steindecriskiv} can be rewritten as
$
\mathcal{E}(s-\underline{s})+\widehat{\mathcal{E}}_o = \mathcal{E}(s-\underline{s})+
D^{-1}\sum_{\ell=1}^L (\widehat{s}_\ell-\underline{\widetilde{r}}_\ell)^2
+\widehat{\Delta}
$,
where
$
\underline{\widetilde{r}}_\ell
=  \langle \underline{\widetilde{r}},\widetilde{\varphi}_\ell \rangle.
$
Thanks to \eqref{eq:estsDelta}, the observable part of the risk estimate can be expressed as
\begin{equation}
\widehat{\mathcal{E}}_o  = \frac{1}{D}\sum_{\ell=1}^L (\widehat{s}_\ell-\underline{\widetilde{r}}_\ell)^2
+ \frac{2\gamma}{D} \sum_{\ell=1}^L \Theta_\ell'(r_\ell) \overline{\gamma}_\ell
-\frac{\gamma}{D} \sum_{\mathbf{p}\in \mathbb{Q}} | H(\mathbf{p}) |^{-2}
.
\label{eq:critorth}
\end{equation}
where $\big(\overline{\gamma}_\ell\big)_{1\le \ell \le L}$ is
given by \eqref{eq:devintertermfni}.

Let us now assume that the coefficients $(r_\ell)_{1\le\ell\le L}$ are classified according to $M\in \mathbb{N}^*$ distinct nonempty index subsets $\mathbb{K}_m$, $m\in \{1,\ldots,M\}$. We have then $L = \sum_{m=1}^M K_m$ where, for every $m\in \{1,\ldots,M\}$, $K_m = \mathrm{card}(\mathbb{K}_m)$.
For instance, for a wavelet decomposition, these subsets may correspond to the subbands associated with the different resolution levels, orientations,... In addition, consider that, for every $m \in \{1,\ldots,M\}$, the estimating 
functions $\big(\Theta_{\ell}\big)_{\ell \in \mathbb{K}_m}$
belong to a given class of parametric functions
and they are characterized by a vector parameter $\boldsymbol{a}_{m}$.
The same estimating function is thus employed for a given subset  $\mathbb{K}_m$ of indices. Then, it can be noticed that the criterion to be minimized in \eqref{eq:critorth} is the sum of $M$ partial MSEs corresponding to each subset $\mathbb{K}_m$. Consequently, we can separately adjust the vector $\boldsymbol{a}_{m}$, for every $m\in \{1,\ldots,M\}$,
so as to minimize
\begin{equation}
\sum_{\ell\in \mathbb{K}_m} \big(\Theta_\ell(r_\ell)-\underline{\widetilde{r}}_\ell\big)^2+ 2\gamma \sum_{\ell\in \mathbb{K}_m}
\Theta_\ell'(r_\ell) \overline{\gamma}_\ell.
\label{eq:critorth2}
\end{equation}

\subsection{Example of LET functions}
 \label{subse:LET}
 As in the previous section, we assume that the coefficients $(r_{\ell})_{1\le \ell\le L}$ as defined in \eqref{eq:modrl} are classified according to $M\in \mathbb{N}^*$ distinct index subsets $\mathbb{K}_m$, $m\in \{1,\ldots,M\}$. Within each class $\mathbb{K}_m$,
a LET estimating function is built from a linear combination of $I_m\in \NN^*$ given functions 
$f_{m,i}\colon\RR \to \RR$ applied to $r_{\ell}$.  So, for every $m \in \{1,\ldots,M\}$ and $\ell \in \mathbb{K}_m$, the estimator takes the form:
\begin{equation}
\Theta_{\ell}(r_{\ell})=\sum_{i=1}^{I_m} a_{m,i}\, f_{m,i}(r_{\ell})
\label{eq:surelet}
\end{equation}
where $(a_{m,i})_{1 \le i \le I_m}$ are scalar real-valued weighting factors. 
We deduce from \eqref{eq:estsu} that the estimate can be expressed as
\begin{equation}
\widehat{s}(\mathbf{x})
=  \sum_{m=1}^M\sum_{i=1}^{I_m} a_{m,i}\,\underline{\beta}_{m,i}(\mathbf{x})
\label{eq:estsbeta}
\end{equation}
where 
\begin{equation}
\underline{\beta}_{m,i}(\mathbf{x}) = \sum_{\ell \in \mathbb{K}_m} f_{m,i}(r_{\ell}) \underline{\widetilde{\varphi}}_{\ell}(\mathbf{x}).
\label{eq:defbetab}
\end{equation}
Then, the problem of optimizing the estimator boils down to the determination of the weights $a_{m,i}$ which minimize the unbiased risk estimate.
According to \eqref{eq:Eo} and \eqref{eq:estsDelta}, this is equivalent to minimize
$
\mathcal{E}( \widehat{s}-  \underline{\widetilde{r}})
+ \frac{2\gamma}{D} \sum_{m=1}^M\sum_{\ell\in \mathbb{K}_m}
\Theta_\ell'(r_\ell) \overline{\gamma}_\ell
$,
where $\big(\overline{\gamma}_\ell\big)_{1\le \ell \le L}$ is
given by \eqref{eq:devintertermfni}.
From \eqref{eq:estsbeta}, it can be deduced that this amounts to minimizing:
\begin{align*}
&\sum\limits_{m_0=1}^M\sum\limits_{i_0=1}^{I_{m_0}}a_{m_0,i_0}
\sum\limits_{m=1}^M\sum\limits_{i=1}^{I_m} a_{m,i}
\langle\underline{\beta}_{m_0,i_0},
  \underline{\beta}_{m,i}\rangle\\
&-2\sum\limits_{m_0=1}^M\sum\limits_{i_0=1}^{I_{m_0}} a_{m_0,i_0}
\langle \underline{\beta}_{m_0,i_0},\underline{\widetilde{r}}\rangle\\
&+2\gamma\sum\limits_{m_0=1}^M \sum\limits_{\ell\in \mathbb{K}_{m_0}}
\sum\limits_{i_0=1}^{I_{m_0}}a_{m_0,i_0}
f'_{m_0,i_0}(r_\ell)\,\overline{\gamma}_\ell.
\end{align*}
This minimization can be easily shown to yield the following set of linear equations:
\begin{multline}
\forall m_0\in\{1,\ldots,M\},\forall i_0\in\{1,\ldots,I_{m_0}\}, \\
 \sum_{m=1}^M\sum_{i=1}^{I_m}  \langle\underline{\beta}_{m_0,i_0},\underline{\beta}_{m,i}\rangle\, a_{m,i}
\\=\langle \underline{\beta}_{m_0,i_0},\underline{\widetilde{r}}\rangle-\gamma\sum_{\ell \in \mathbb{K}_{m_0}} f'_{m_0,i_0}(r_{\ell})\,\overline{\gamma}_\ell.
\label{eq:estssureletlin}
\end{multline}

\section{Parameter choice}
\label{se:param}
\subsection{Choice of analysis/synthesis functions}
Using the same notations as in Section \ref{se:flexiblity}, let
$\{\mathbb{K}_m, 1 \le m \le M\}$ be a partition of $\{1,\ldots,L\}$.
Consider now a frame of  $\RR^{D_1\times \cdots \times D_d}$,
$\big((\psi_{m,\mathbf{k}_\ell})_{\ell \in \mathbb{K}_m}\big)_{1 \le m \le M}$, where, for every $m\in \{1,\ldots,M\}$,
$\psi_{m,\mathbf{0}}$ is some field in $\RR^{D_1\times \cdots \times D_d}$
and, for every $\ell \in \mathbb{K}_m$, $\psi_{m,\mathbf{k}_\ell}$ denotes its $\mathbf{k}_\ell$-periodically
shifted version where $\mathbf{k}_\ell$ is some shift value
in $\mathbb{D}$. Notice that, by appropriately choosing the sets $(\mathbb{K}_m)_{1 \le m \le M}$,  any frame of $\RR^{D_1\times \cdots \times D_d}$
can be written under this form
but that it is mostly useful to describe periodic wavelet bases, wavelet packets \cite{Coifman92}, mirror wavelet bases \cite{KALIFA_MALLAT_03}, redundant/undecimated wavelet representations as well as related frames 
\cite{SEL_BARA_KING_05,CC_DUVAL_JCP_06,Mndo05,Mallat08}.
For example, for a classical 1D periodic wavelet basis, 
$M-1$ represents the number of resolution levels
and, for every $\ell$ in subband $\mathbb{K}_m$ at resolution level $m\in \{1,\ldots, M-1\}$, the shift parameter $\mathbf{k}_\ell$ is a multiple of $2^m$ ($\mathbb{K}_M$ being here the index subset related to the approximation subband).

A possible choice for the analysis family $(\varphi_\ell)_{1 \le \ell \le L}$ is then obtained by setting 
\begin{align}
\forall m \in \{1,\ldots,M\}&, \forall \ell \in \mathbb{K}_m,
\forall \mathbf{p}\in \mathbb{D},\quad\nonumber\\
\Phi_\ell(\mathbf{p})&= 
G(\mathbf{p})\Psi_{m,\mathbf{k}_\ell}(\mathbf{p})\nonumber\\
&= \exp(-2\pi \imath
{\mathbf{k}_\ell}^\top{\boldsymbol D}^{-1}{\mathbf{p}}) 
G(\mathbf{p})\Psi_{m,\mathbf{0}}(\mathbf{p})
\label{eq:Phiellden}
\end{align}
where $G(\mathbf{p})$ typically corresponds to the frequency response of an ``inverse'' of the degradation filter.
It can be noticed that a similar choice is made in
the WaveD estimator \cite{PICARD_04} by setting, for every $\mathbf{p} \in \mathbb{Q}$,
$G(\mathbf{p}) = 1/\big(H(\mathbf{p})\big)^*$ (starting from a dyadic Meyer wavelet basis).
By analogy with Wiener filtering techniques, a more general form for the frequency response of this filter can be chosen:
\begin{equation}
G(\mathbf{p}) = \frac{H(\mathbf{p})}{|H(\mathbf{p})|^2+
\lambda}
\label{eq:genG}
\end{equation}
where $\lambda \ge 0$. Note that, due to \eqref{eq:Phiellden}, the 
computation of the coefficients $(r_\ell)_{1 \le \ell \le L}$ 
amounts to the computation of the frame coefficients:
\begin{equation}
\forall m \in \{1,\ldots,M\}, \forall \ell \in \mathbb{K}_m,\qquad
r_\ell = \langle \check{r},\psi_{m,\mathbf{k}_\ell}\rangle
\end{equation}
where $\check{r}$ is the field with discrete Fourier coefficients
\begin{equation}
\check{R}(\mathbf{p})
= \begin{cases}
\displaystyle \big(G(\mathbf{p})\big)^* R(\mathbf{p})  & \mbox{if $\mathbf{p} \in \mathbb{Q}$}\\
0 & \mbox{otherwise.}
\end{cases}
\end{equation}

Concerning the associated synthesis family
$(\widetilde{\varphi}_\ell)_{1\le \ell \le L}$, we simply choose the dual synthesis
frame of $\big((\psi_{m,\mathbf{k}_\ell})_{\ell \in \mathbb{K}_m}\big)_{1 \le m \le M}$which, with a slight abuse of notation,  will be assumed of the form:\linebreak
$\big((\widetilde{\varphi}_{m,\mathbf{k}_\ell})_{\ell \in \mathbb{K}_m}\big)_{1 \le m \le M}$
where, for every $\ell \in \mathbb{K}_m$, $\widetilde{\varphi}_{m,\mathbf{k}_\ell}$
denotes the $\mathbf{k}_\ell$-periodically shifted version of 
$\widetilde{\varphi}_{m,\mathbf{0}}$.
So, basically the restoration method can be summarized by Fig. \ref{fig:method}.

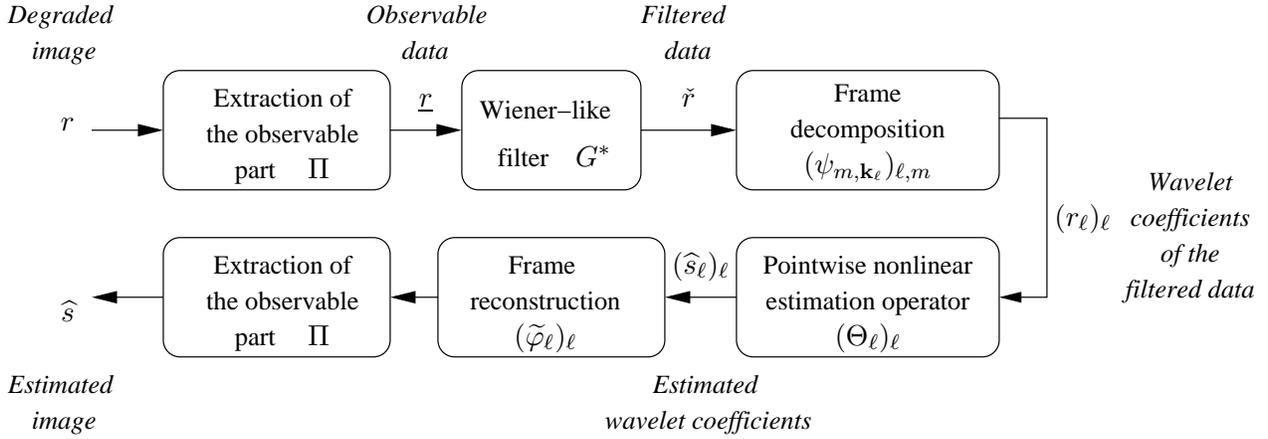
\begin{figure*}[h!tb]
\begin{center}
 \input{method.pstex_t}
\end{center}
\caption{Restoration method.
\label{fig:method}}
\end{figure*}

With these choices, it can be deduced from \eqref{eq:devintertermfni} that 
\begin{multline}
\forall m \in \{1,\ldots,M\}, \forall \ell \in \mathbb{K}_m,\qquad\\
\overline{\gamma}_\ell = \frac{1}{D} 
\sum_{\mathbf{p} \in \mathbb{Q}} \frac{\Psi_{m,\mathbf{0}}(\mathbf{p})
\big(\widetilde{\Phi}_{m,\mathbf{0}}(\mathbf{p})\big)^*}
{|H(\mathbf{p})|^2+ \lambda}.
\end{multline}
This shows that only $M$ values of $\overline{\gamma}_\ell$ need to
be computed (instead of $L$). Similarly, simplified forms of the constants 
$(\overline{\gamma}_{\ell,i})_{ 1\le \ell,i \le L}$ and $(\kappa_\ell)_{1 \le \ell
\le L}$ as defined by \eqref{eq:nf5iii} and \eqref{eq:defkappa} can be easily obtained.

\subsection{Choice of estimating functions} \label{se:choiceest}
We will employ LET estimating functions due to the simplicity of their optimization, as explained in Section~\ref{subse:LET}. More precisely, the following 
two possible forms will be investigated in this work:
\begin{itemize}
\item nonlinear estimating function in \cite{BLU_LUISIER_07}: we set $I_m = 2$,
take for $f_{m,1}$ the identity function and choose
\begin{equation}
\forall \rho \in \mathbb{R},
f_{m,2}(\rho) = \left(1-\exp\Big(-\frac{\rho^8}{(\omega \sigma_m)^8}\Big)\right)\rho
\label{eq:Blu}
\end{equation}
where $\omega \in ]0,\infty[$ 
 and $\sigma_m$ is the standard deviation
of $(n_\ell)_{\ell \in \mathbb{K}_m}$. According to \eqref{eq:redefnl} and \eqref{eq:Phiellden},
we have, for any $\ell \in \mathbb{K}_m$,
$
\sigma_m^2 = \gamma D^{-1} \sum_{\mathbf{p}\in\mathbb{Q}}|\Phi_\ell(\mathbf{p})|^2 = \gamma D^{-1} \sum_{\mathbf{p}\in\mathbb{Q}}|G(\mathbf{p})|^2|\Psi_{m,\mathbf{0}}(\mathbf{p})|^2$.

\item nonlinear estimating function in \cite{PES_LEPO_97}: again, we set $I_m =2$,
and take for $f_{m,1}$ the identity function but, we choose:
\begin{align}
&\forall \rho \in \mathbb{R},\qquad\nonumber\\
&f_{m,2}(\rho) = \left(\tanh\Big(\frac{\rho+\xi\sigma_m}{\omega'\sigma_m}\Big)
- \tanh\Big(\frac{\rho-\xi\sigma_m}{\omega'\sigma_m}\Big)\right)\rho
\label{eq:sigmo}
\end{align}
where $(\xi,\omega')\in ]0,\infty[^2$
and $\sigma_m$ is defined as for the previous
estimating function.
\end{itemize}

\section{Experimental results}
\label{sec:simuls}
\subsection{Simulation context}
In our experiments, the test data set contains six 8-bit images of size 512 $\times$ 512 which are displayed in Fig. \ref{fig:orig}. 
Different convolutions have been applied:
(i) $5 \times 5$ and $7 \times 7$ uniform blurs,
(ii) Gaussian blur with standard deviation $\sigma_h$ equal to $2$,
(iii) cosine blur defined by: $\forall (p_1,p_2) \in \{0,\ldots,D_1-1\}
\times \{0,\ldots,D_2-1\}$, $H(p_1,p_2) = H_1(p_1) H_2(p_2)$
where
\begin{multline}
\forall i \in \{1,2\},\qquad\\
H_i(p_i) = \begin{cases}
1 & \mbox{if $0\le p_i \le F_c D_i$}\\
\displaystyle
\cos\Big(\frac{\pi(p_i-F_c D_i)}{(1-2F_c)D_i}\Big)& \mbox{if $F_c D_i \le p_i \le D_i/2$}\\
\big(H_i(D_i-p_i)\big)^* & \mbox{otherwise}
\end{cases} 
\end{multline}
with $F_c \in [0,1/2)$,
(iv) Dirac (the restoration problem then reduces to a 
denoising problem)
and, realizations of a zero-mean white Gaussian noise have been added to the blurred images. The noise variance $\gamma$ is chosen so that the averaged blurred signal to noise ratio $\mathrm{BSNR}$  reaches a given target value, where
$
 \mathrm{BSNR}\,\eqdef\, 10\log_{10}\left(\parallel \widetilde{h}*s\parallel^2/(D\gamma)\right)$.
The performance of a restoration method is measured by the averaged Signal to
Noise Ratio:
$
\mathrm{SNR}\,\eqdef\, 
10\log_{10}\left(\widehat{\E}[s^2]/\widehat{\E}[(s-\widehat{s})^2]\right)$
where $\widehat{\E}$ denotes the spatial average operator.
In our simulations, we have chosen the set $\mathbb{P}$ as given by
\eqref{eq:defPchi} where the threshold value $\chi$ 
is automatically adjusted so as to secure a reliable estimation
of the risk while maximizing the size of the set $\mathbb{Q}$.
In practice, $\chi$ has been set, through a dichotomic search, 
to the smallest positive value such that
$\widehat{\mathcal{E}}_o > 10 \sqrt{V_{\mathrm{max}}}$,
where $V_{\mathrm{max}}$ is an upper bound of 
$\mathsf{Var}[\mathcal{E}_o-\widehat{\mathcal{E}}_o]$.
This bound has been derived from \eqref{eq:varsteinf} under
some simplifying assumptions aiming at facilitating
its computation. In an empirical manner, the parameter $\lambda$ in \eqref{eq:genG} has been chosen
proportional to the ratio of the noise variance to the variance of the blurred
image, by taking $\lambda = 3 \gamma/(\widehat{\E}[r^2]-(\widehat{\E}[r])^2-\gamma)$. The other parameters of the method have been set to 
$\omega = 3$ in \eqref{eq:Blu} and $(\xi,\omega') = (3.5,2.25)$ 
in \eqref{eq:sigmo}.

To validate our approach, we have made comparisons with
state-of-the-art wavelet-based restoration methods and some other restoration approaches. For all these methods, 
symlet-8 wavelet decompositions performed over 4 resolution levels have been used  \cite{Daub92}. 
The first approach is the ForWaRD method\footnote{A Matlab toolbox can be downloaded from \textsf{http://www.dsp.rice.edu/software/ward.shtml}.} which employs a translation invariant wavelet representation \cite{NASON_SILV_95,Pesq96}. The ForWaRD estimator has been applied with an optimized value of the regularization parameter. 
The same translation invariant wavelet decomposition is used for the proposed SURE-based method.
The second method we have tested is the TwIST\footnote{A Matlab toolbox can be downloaded from \textsf{http://www.lx.it.pt/$\sim$bioucas/code.htm}.} algorithm \cite{BIOUCAS_FIGUEIREDO_07} considering a total variation penalization term.
The third approach is the variational method in \cite[Section 6]{CC_PES_COMB_07} (which extends the method in \cite{DEMOL_05}) where we use
a tight wavelet frame consisting of the union of four shifted
orthonormal wavelet decompositions. The shift parameters are $(0,0)$, $(1,0)$, $(0,1)$ and $(1,1)$. We have also included in our comparisons the results
obtained with the classical Wiener filter and with a least squares optimization
approach using a Laplacian regularization operator.\footnote{We use
the implementations of these methods provided in the Matlab Image Processing
Toolbox, assuming that the noise level is known.}

\begin{figure*}[h!tb]
\begin{center}
\begin{tabular}{cccccc}
\includegraphics[width=2.5cm]{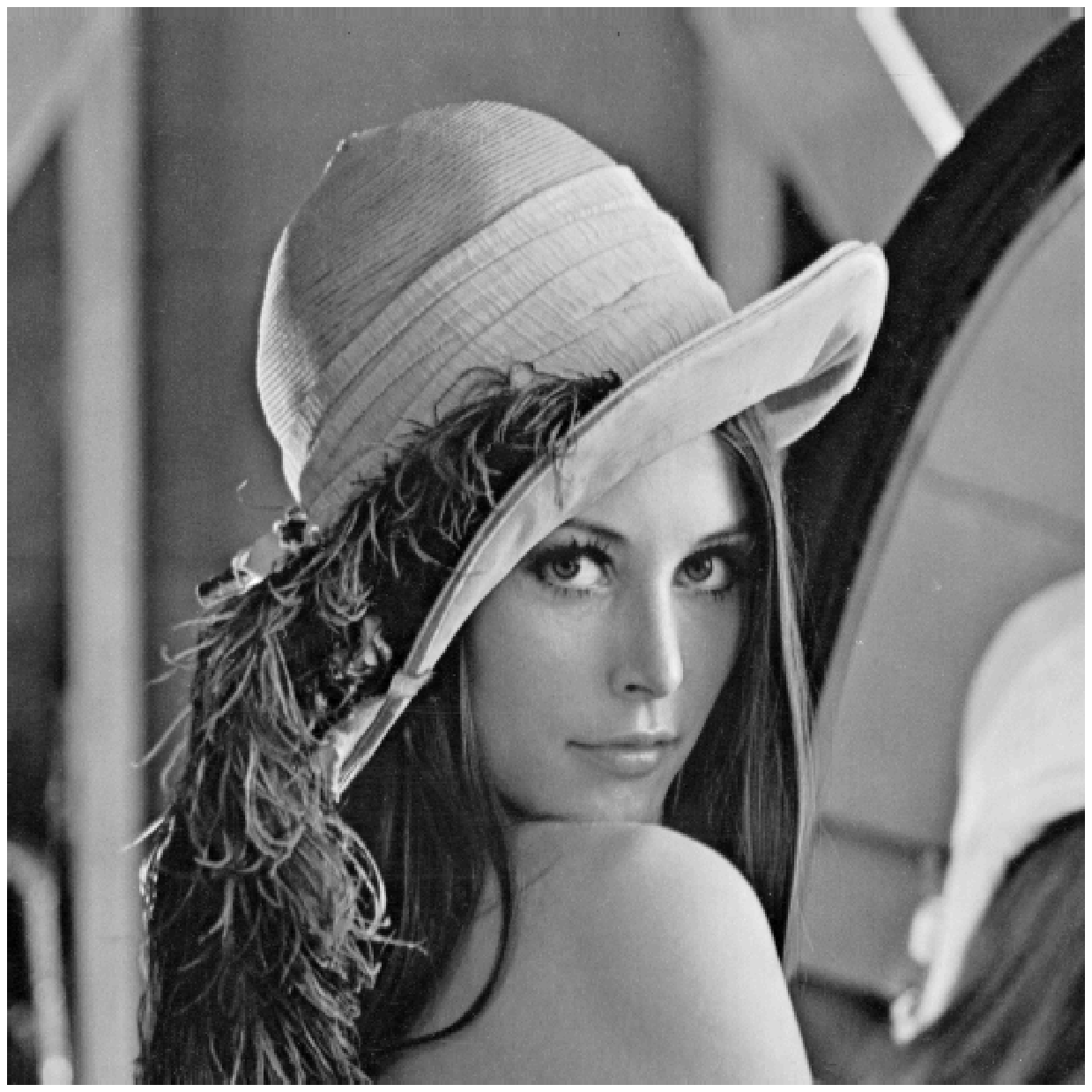} & \includegraphics[width=2.5cm]{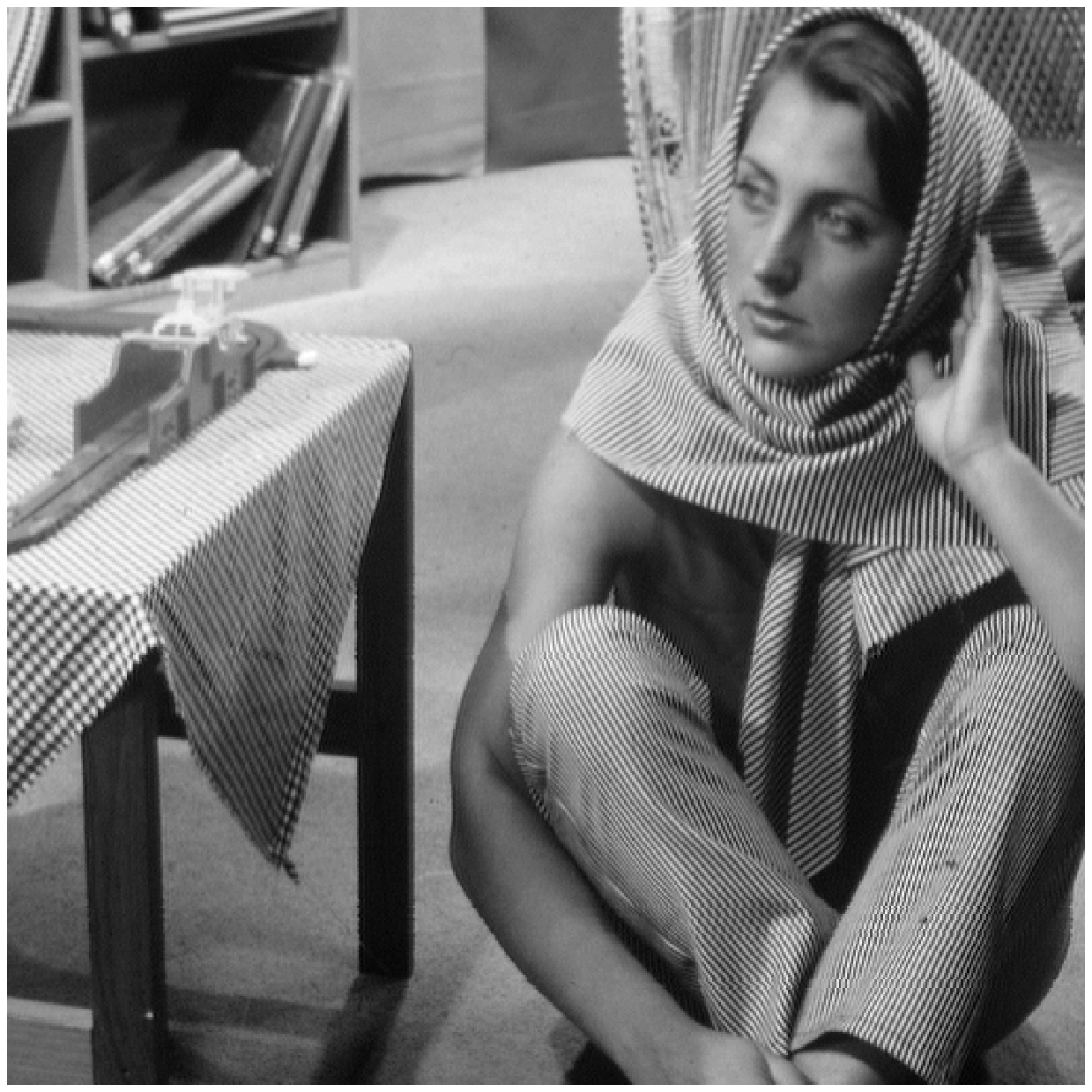} & \includegraphics[width=2.5cm]{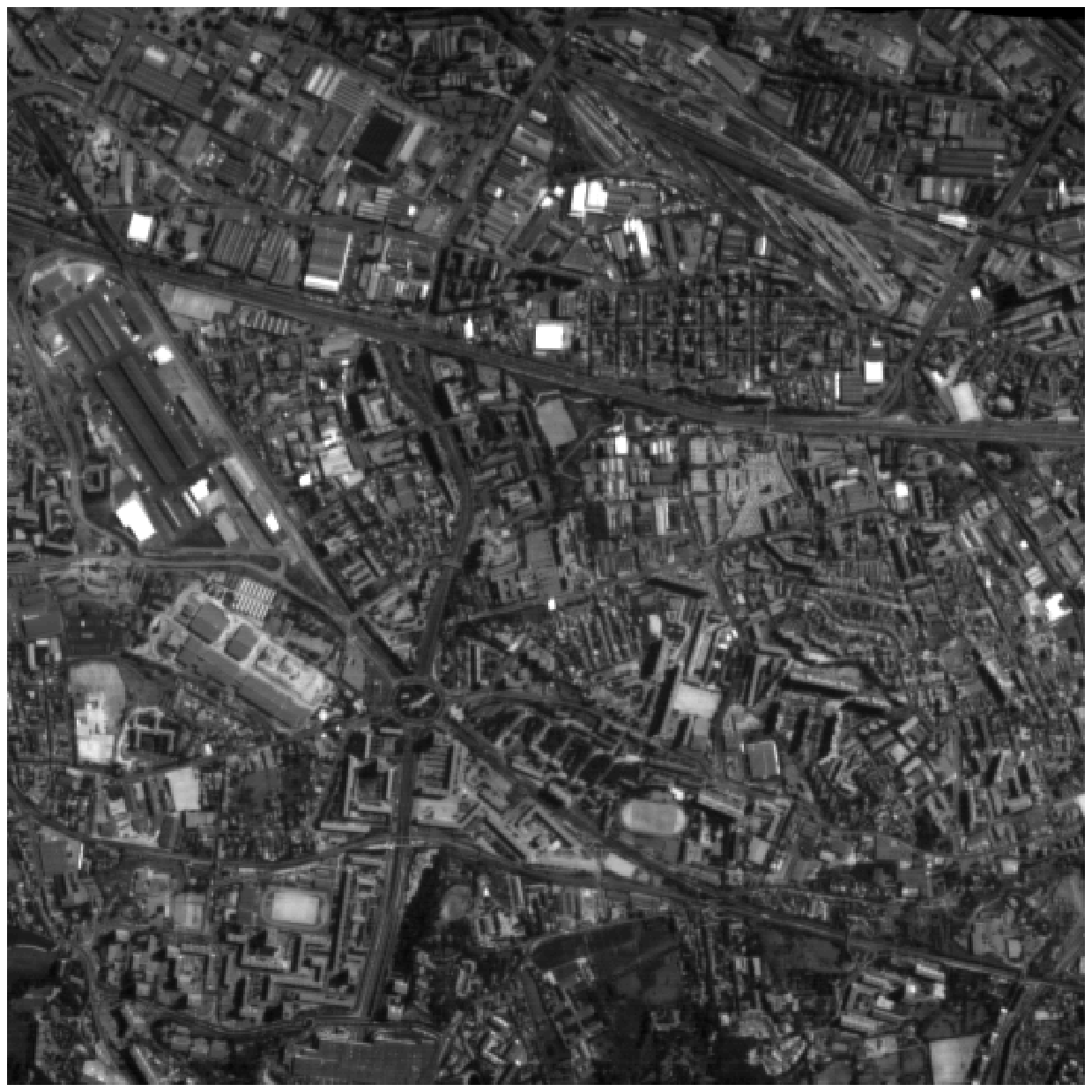}
& \includegraphics[width=2.5cm]{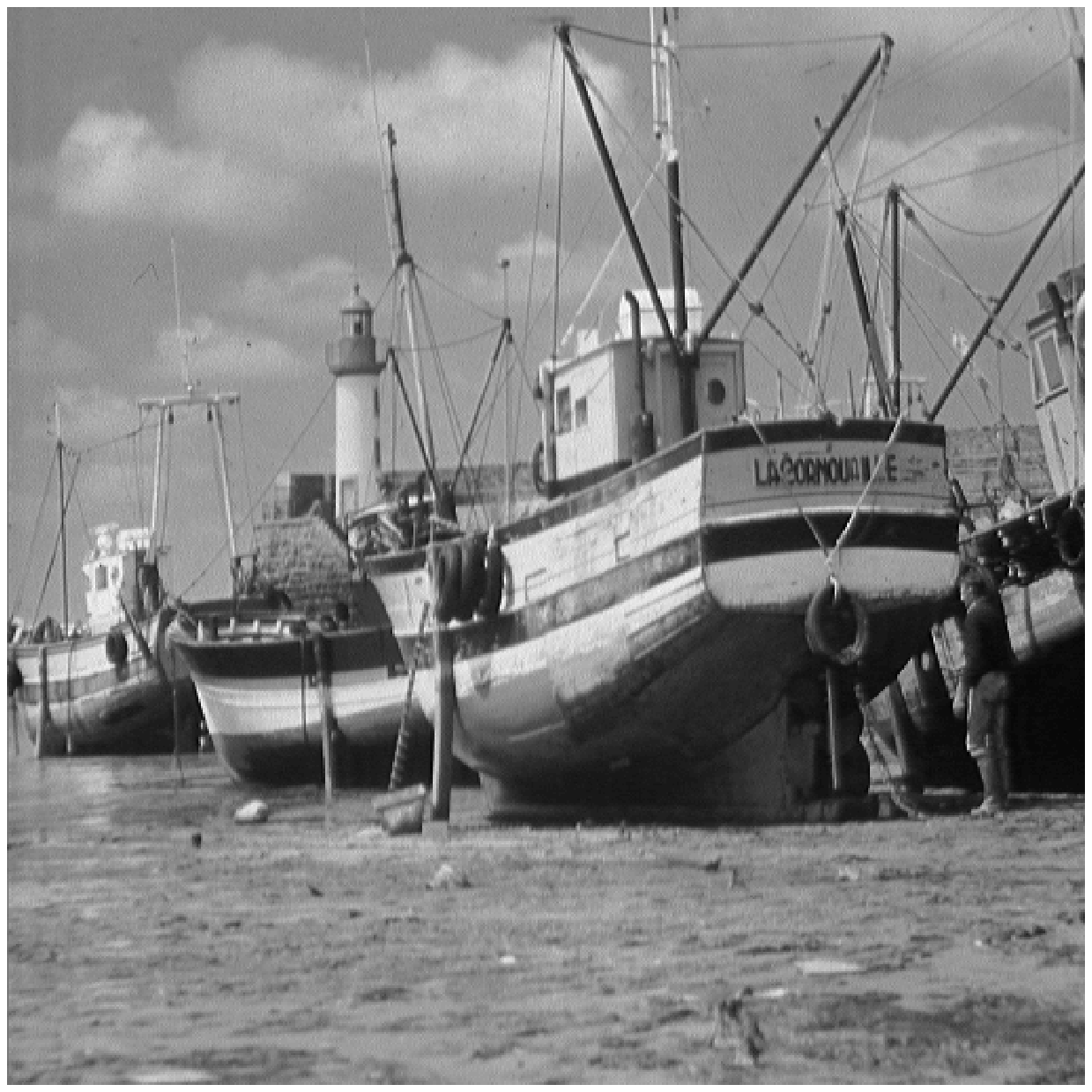} & \includegraphics[width=2.5cm]{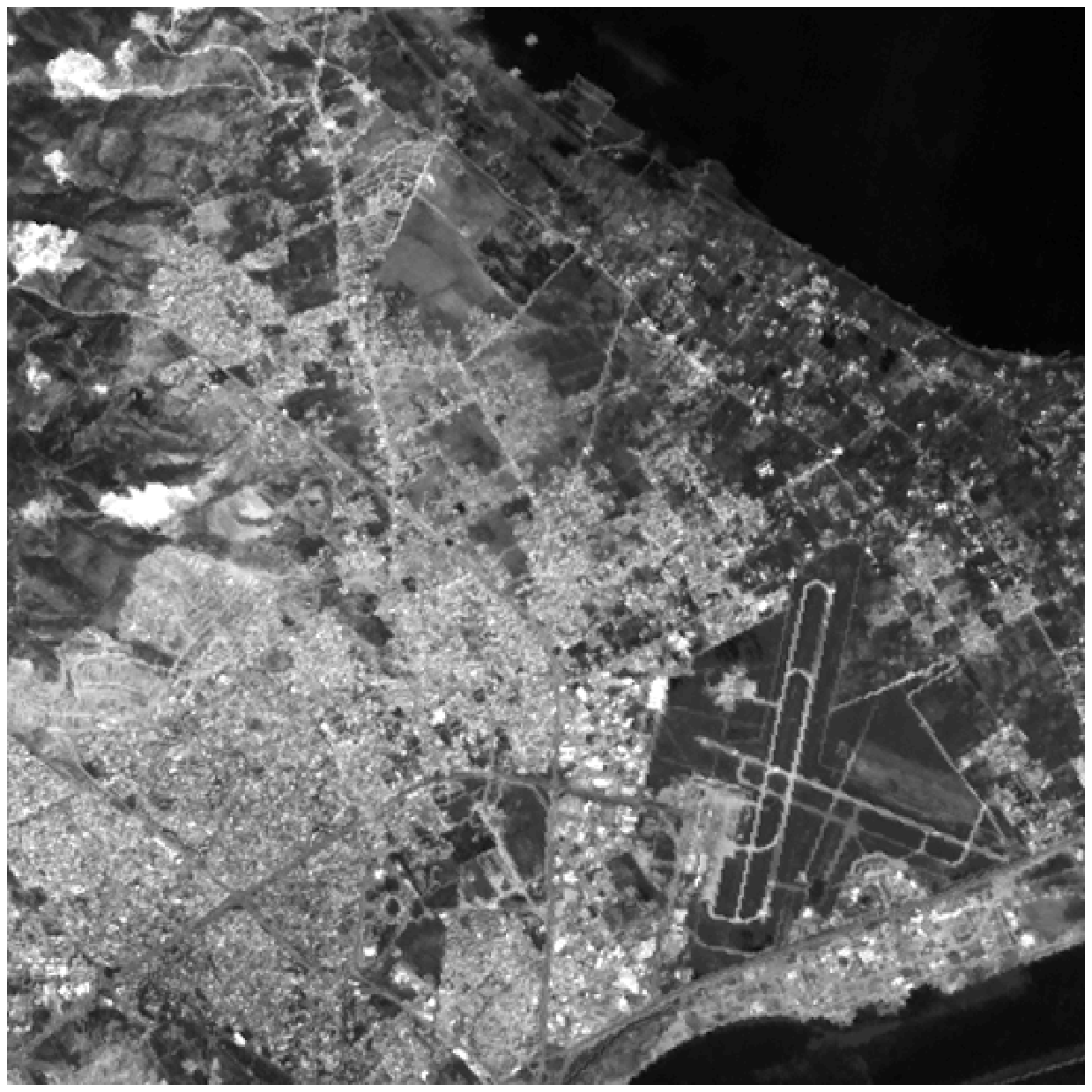} & \includegraphics[width=2.5cm]{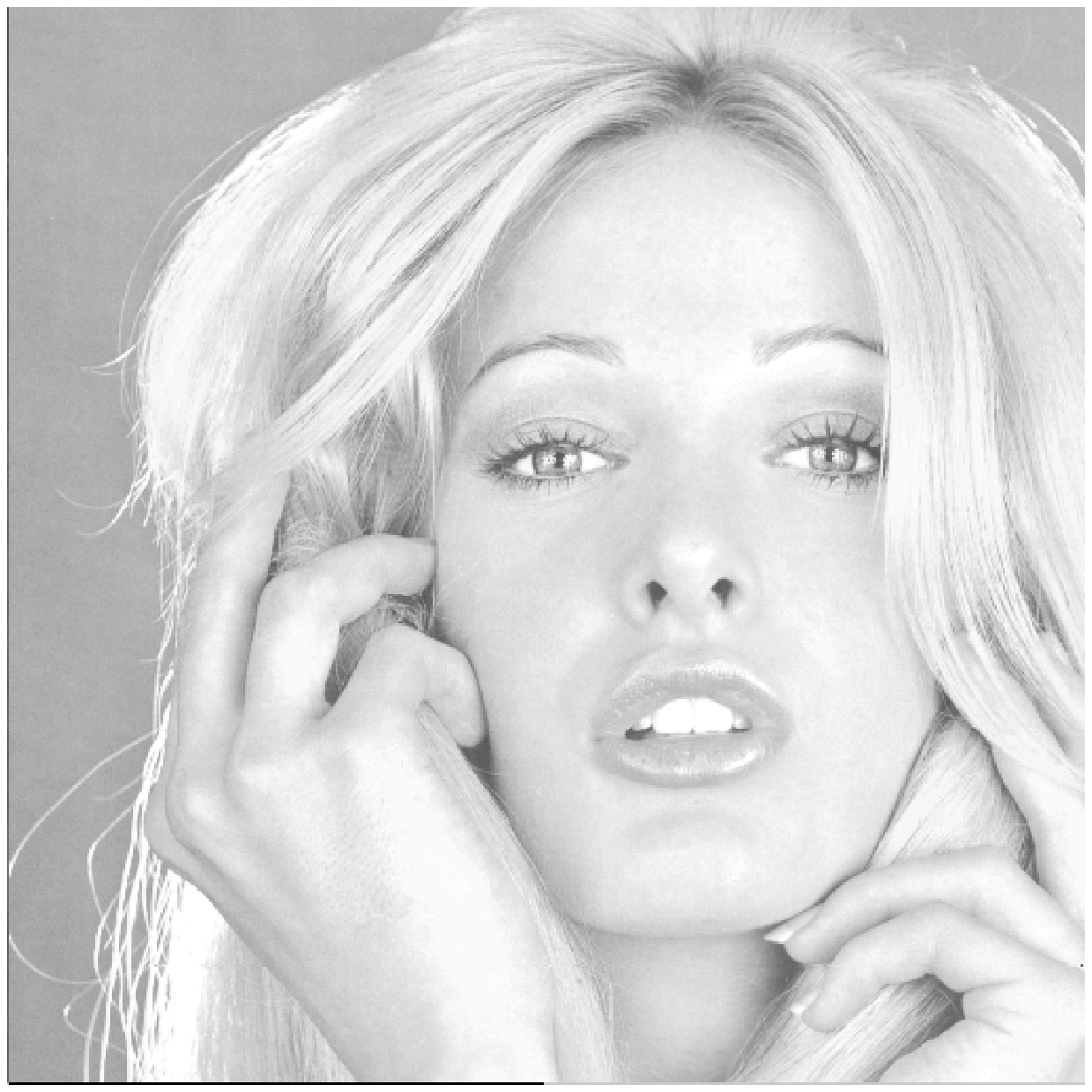} \\
(a) Lena & (b) Barbara & (c) Marseille &
(d) Boat & (e) Tunis & (f) Tiffany\\
\end{tabular}
\end{center}
\caption{Original images: (a) Lena, (b) Barbara, (c) Marseille, (d) Boat, (e) Tunis and (f) Tiffany. \label{fig:orig}}
\end{figure*}
\subsection{Numerical results}
Table \ref{tab:images_unif} provides the values of the SNR achieved by the different considered techniques for several values of the BSNR and a given form of blur (uniform $5 \times 5$) on the six test images.
All the provided quantitative results are median values computed over 10 noise
realizations.
It can be observed that, whatever the considered image is, 
SURE-based restoration methods generally lead to significant gains
w.r.t. the other approaches, especially for low BSNRs.
Furthermore, the two kinds of nonlinear estimating function which have been evaluated lead to almost identical results.
It can also be noticed that the ForWaRD and TwIST methods perform quite well
in terms of MSE for high BSNR.
However, by examining more carefully the restored images,
it can be seen that these methods may better recover uniform areas, at the expense of a loss
of some detail information which is better preserved by the considered SURE-based
method. This behaviour is 
visible on Fig. \ref{fig:rest_im}
where the proposed approach allows us to better recover Barbara's stripe trouser.

Table \ref{tab:tunis_blur} provides the SNRs obtained with the different techniques for several values of the $\mathrm{BSNR}$ and various blurs on Tunis image (see Fig. \ref{fig:orig} (e)). The reported results
allow us to confirm the good performance of SURE-based methods.
The lower performance of the wavelet-based variational approach may be related to the fact that it requires the estimation of the hyperparameters of the prior distribution of the wavelet 
coefficients.
This estimation has been performed by a maximum likelihood approach which is suboptimal in terms of mean square restoration error. The results at the bottom-right
of Table \ref{tab:tunis_blur} are in agreement with those
in \cite{BLU_LUISIER_07,RAPHAN_SIMONC_08} showing
the outperformance of LET estimators for denoising problems.
The poorer results obtained with ForWaRD in this case indicate that this method is tailored for deconvolution problems.

In the previous experiments, for all the considered methods, the noise variance $\gamma$ was assumed to be known.
Table \ref{tab:tunis_blur_est} gives the SNR values obtained with the different techniques for several noise levels and various blurs on Tunis image, when the noise variance is estimated via the classical median absolute deviation (MAD) wavelet estimator \cite[p. 447]{MALLAT_98}. One can observe that the results are close to  the case when the noise variance is known, except when the problem reduces to a denoising problem associated with a high BSNR. In this case indeed, the MAD estimator does not provide a precise estimation of the noise variance.
However, the restoration results are still satisfactory.

\begin{table*}[h!tb]
\caption{Restoration results for a $5 \times 5$ uniform blur:
initial SNR ($\mathrm{SNR\_i}$) and SNR obtained with our approach
using the nonlinear function in \eqref{eq:Blu}
($\mathrm{SNR\_b}$), our approach using the nonlinear function in 
\eqref{eq:sigmo} ($\mathrm{SNR\_s}$),  ForWaRD
($\mathrm{SNR\_f}$), TwIST ($\mathrm{SNR\_t}$), the wavelet-based variational approach ($\mathrm{SNR\_v}$), the Wiener filter ($\mathrm{SNR\_w}$) and the regularized quadratic method~($\mathrm{SNR\_r}$).
\label{tab:images_unif}}
\begin{center}
\begin{tabular}{|c||c||c|c|c|c|c||c||c||c|c|c|c|c|}
\hline
Image & BSNR & 10 & 15 &  20 & 25 & 30 & Image & BSNR & 10 & 15 &  20 & 25 & 30\\
\hline
\hline
& SNR\_i & 9.796 & 14.26 & 17.89 & 20.21 & 21.31  &  & SNR\_i & 9.713 & 14.03 & 17.39 & 19.39 & 20.28 \\
 & SNR\_b & \textbf{19.98} & \textbf{21.36} & 22.43 & 23.48 & 24.62 & & SNR\_b & \textbf{19.09} & \textbf{20.27} & \textbf{21.31} & \textbf{22.37} & 23.59 \\
 & SNR\_s & 19.93 & 21.29 & 22.39 & 23.45 & 24.58 & & SNR\_s & 19.05 & 20.22 & 21.27 & 22.36 & 23.57 \\
Lena& SNR\_f & 18.27 & 20.04 & 21.36 & 23.35 & 24.63 & Boat & SNR\_f & 16.75 & 19.05 & 20.68 & 22.14 & 23.49 \\ 
 & SNR\_t & 19.52 & 21.17 & \textbf{22.79} & \textbf{23.90} & \textbf{24.91} & & SNR\_t & 18.29 & 19.67 & 21.09 & 22.35 &  \textbf{23.80} \\
& SNR\_v & 17.91 & 20.01 & 21.34 & 22.41 & 23.42  & & SNR\_v & 14.37 & 17.46 & 19.31 & 20.54 & 21.88 \\
& SNR\_w & 15.82 & 19.69 & 21.54 & 22.23 & 22.46  & & SNR\_w & 15.80 & 19.19 & 20.56 & 21.02 & 21.17 \\
 & SNR\_r & 18.70 & 20.06 & 21.25 & 22.18 & 22.43  & & SNR\_r & 18.18 & 19.31 & 20.35 & 20.96 & 21.15 \\
\hline
\hline
 & SNR\_i & 9.366 & 13.10 & 15.54 & 16.72 & 17.17 & & SNR\_i & 9.713 & 14.03 & 17.37 & 19.36 & 20.24 \\
  & SNR\_b & \textbf{17.02} & \textbf{17.54} & 18.05 & \textbf{18.79} & \textbf{19.75} & & SNR\_b & \textbf{18.63} & \textbf{19.73} & \textbf{20.75} & \textbf{21.72} & \textbf{22.66} \\
  & SNR\_s & 16.99 & 17.52 & \textbf{18.06} & 18.77 & 19.63 & & SNR\_s & 18.62 & \textbf{19.73} & 20.73 & 21.70 & \textbf{22.66} \\
Barbara & SNR\_f & 16.14 & 17.04 & 17.62 & 18.56 & 19.49 & Tunis & SNR\_f & 16.57 & 18.54 & 19.99 & 21.20 & 22.29 \\
 & SNR\_t & 16.74 & 17.45 & 17.95 & 18.38 & 19.07 & & SNR\_t & 18.03 & 18.94 & 20.35 & 21.50 & 22.56 \\
 & SNR\_v & 16.41 & 17.26 & 17.76 & 18.32 & 18.94  & & SNR\_v & 17.45 & 18.73 & 19.60 & 20.54 & 21.56 \\
 & SNR\_w & 14.53 & 16.92 & 17.78 & 18.02 & 18.10 & & SNR\_w & 15.80 & 19.07 & 20.37 & 20.79 & 20.92 \\
 & SNR\_r & 16.53 & 17.11 & 17.58 & 17.92 & 18.10 & & SNR\_r & 18.13 & 19.20 & 20.17 & 20.76 & 20.91\\
 \hline
\hline
 & SNR\_i & 8.926 & 11.93 & 13.57 & 14.25 & 14.49 & & SNR\_i & 9.923 & 14.73 & 19.15 & 22.72 & 24.95 \\
 & SNR\_b & 13.92 & \textbf{15.12} & \textbf{16.21} & 17.27 & 18.46 & & SNR\_b & \textbf{24.18} & \textbf{25.28} & 26.13 & 26.92 & 28.09 \\
 & SNR\_s & \textbf{13.93} & \textbf{15.12} & \textbf{16.21} & \textbf{17.28} & \textbf{18.48} & & SNR\_s & 24.16 & 25.24 & 26.11 & 26.92 & 28.09 \\ 
Marseille & SNR\_f & 13.10 & 14.42 & 15.75 & 17.00 & 18.32 & Tiffany & SNR\_f & 17.93 & 21.72 & 23.67 & 26.53 & \textbf{28.12} \\ 
 & SNR\_t & 12.74 & 14.19 & 15.52 & 16.87 & 18.36 & & SNR\_t & 23.08 & 25.17 & \textbf{26.23} & \textbf{27.20} & 27.85 \\
 & SNR\_v & 13.57 & 14.99 & 16.12 & 16.92 & 17.66 & & SNR\_v & 21.81 & 24.47 & 25.63 & 26.44 & 27.46 \\
 & SNR\_w & 12.60 & 14.70 & 15.53 & 15.81 & 15.90 & & SNR\_w & 18.01 & 23.13 & 25.48 & 26.28 & 26.53 \\
 & SNR\_r & 13.25 & 14.43 & 15.45 & 15.78 & 15.90 & & SNR\_r & 23.65 & 24.60 & 25.42 & 26.12 & 26.44 \\
 \hline
\end{tabular}\end{center}
\end{table*}



\begin{figure*}[h!tb]
\begin{center}
\begin{tabular}{cc}
\includegraphics[width=3.7cm]{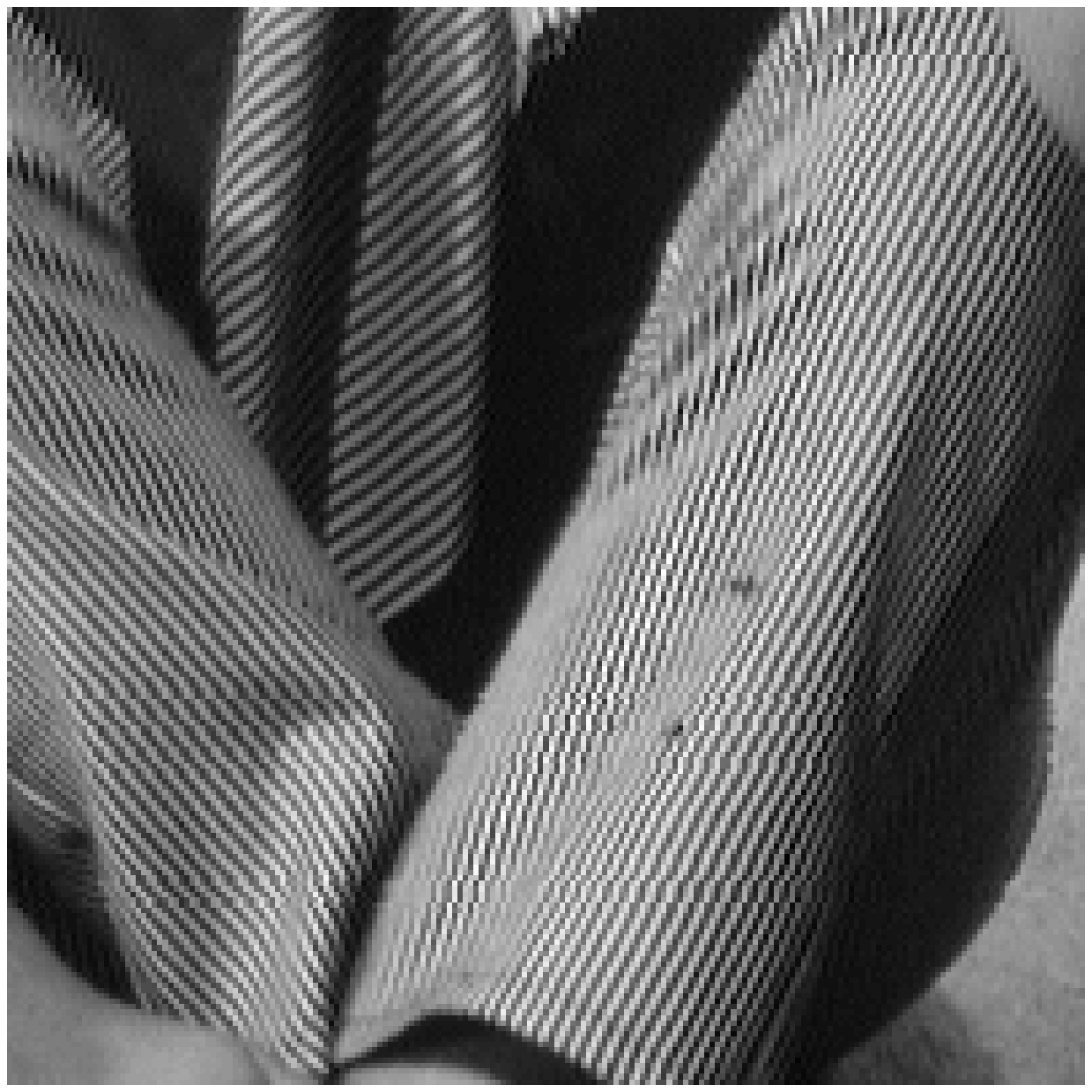}	& \includegraphics[width=3.7cm]{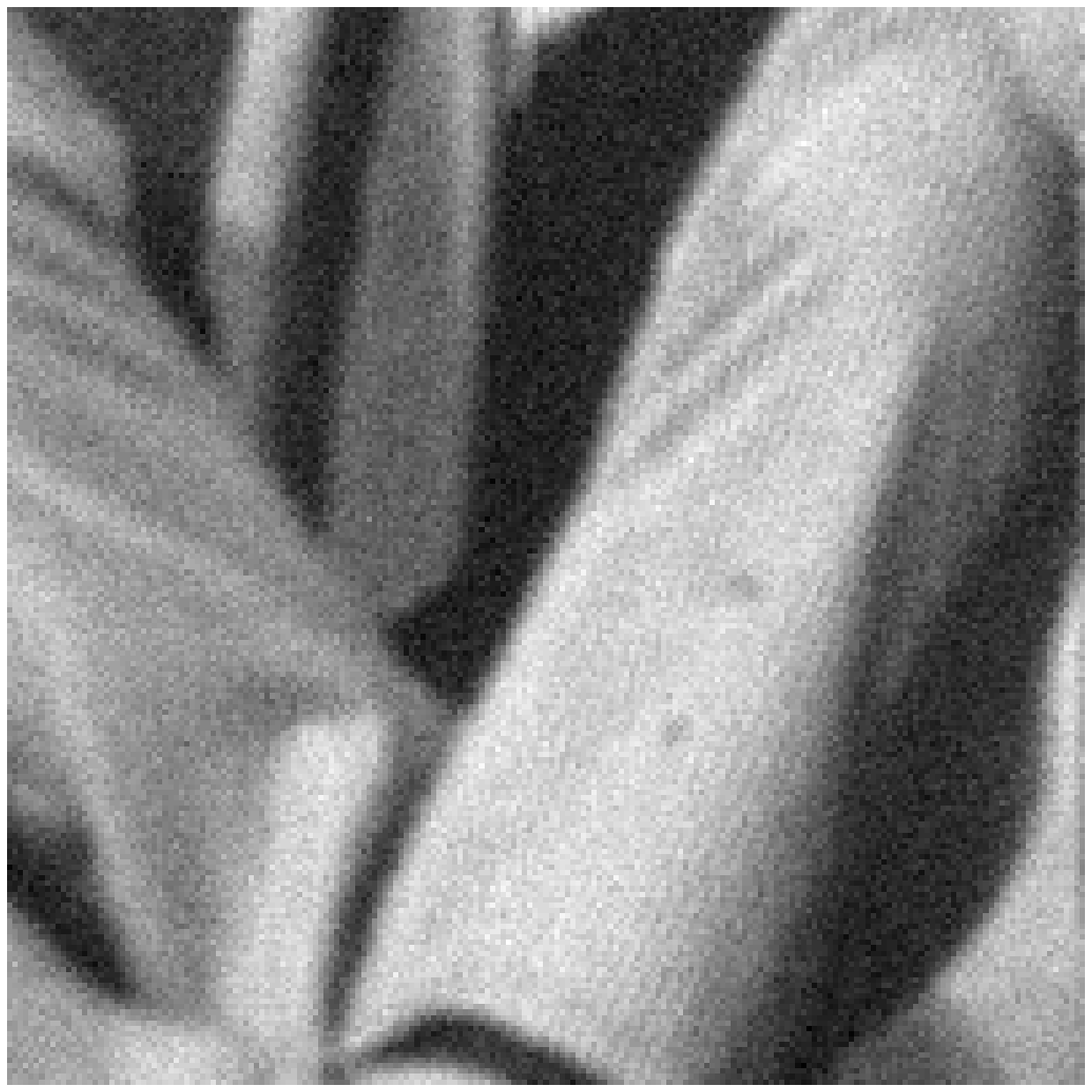} \\
(a) & (b) \\
\end{tabular}
\begin{tabular}{ccc}
\includegraphics[width=3.7cm]{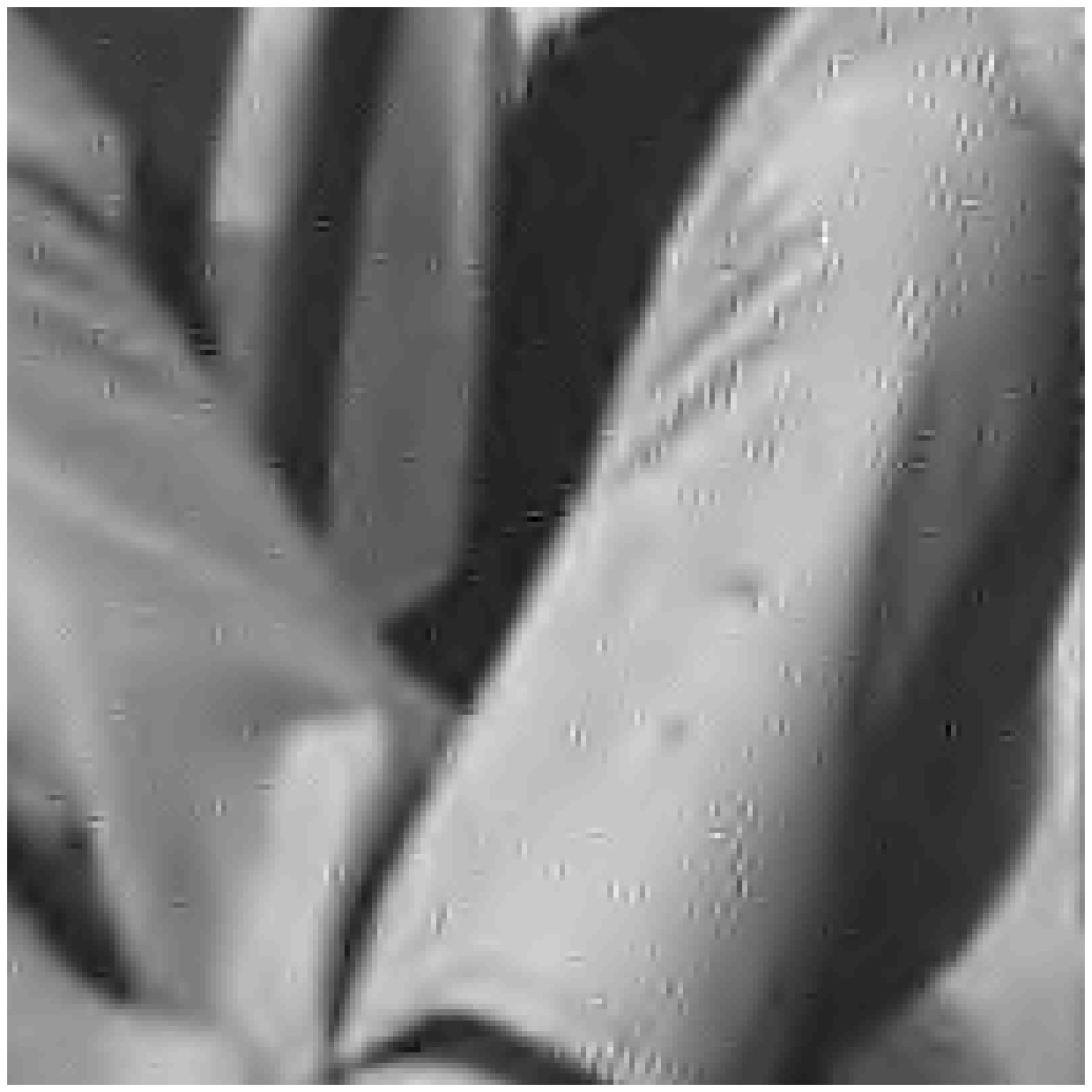}
& \includegraphics[width=3.7cm]{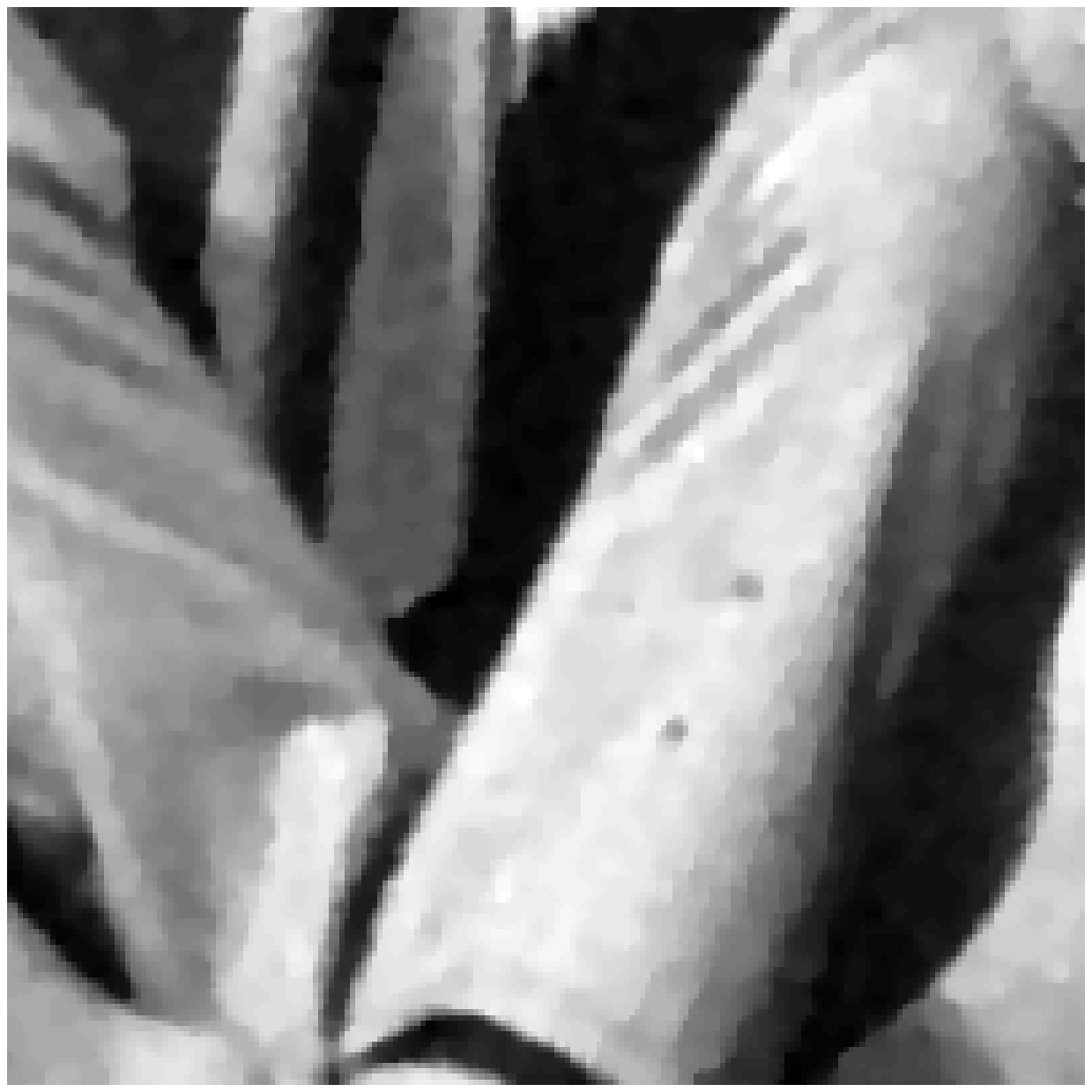}
& \includegraphics[width=3.7cm]{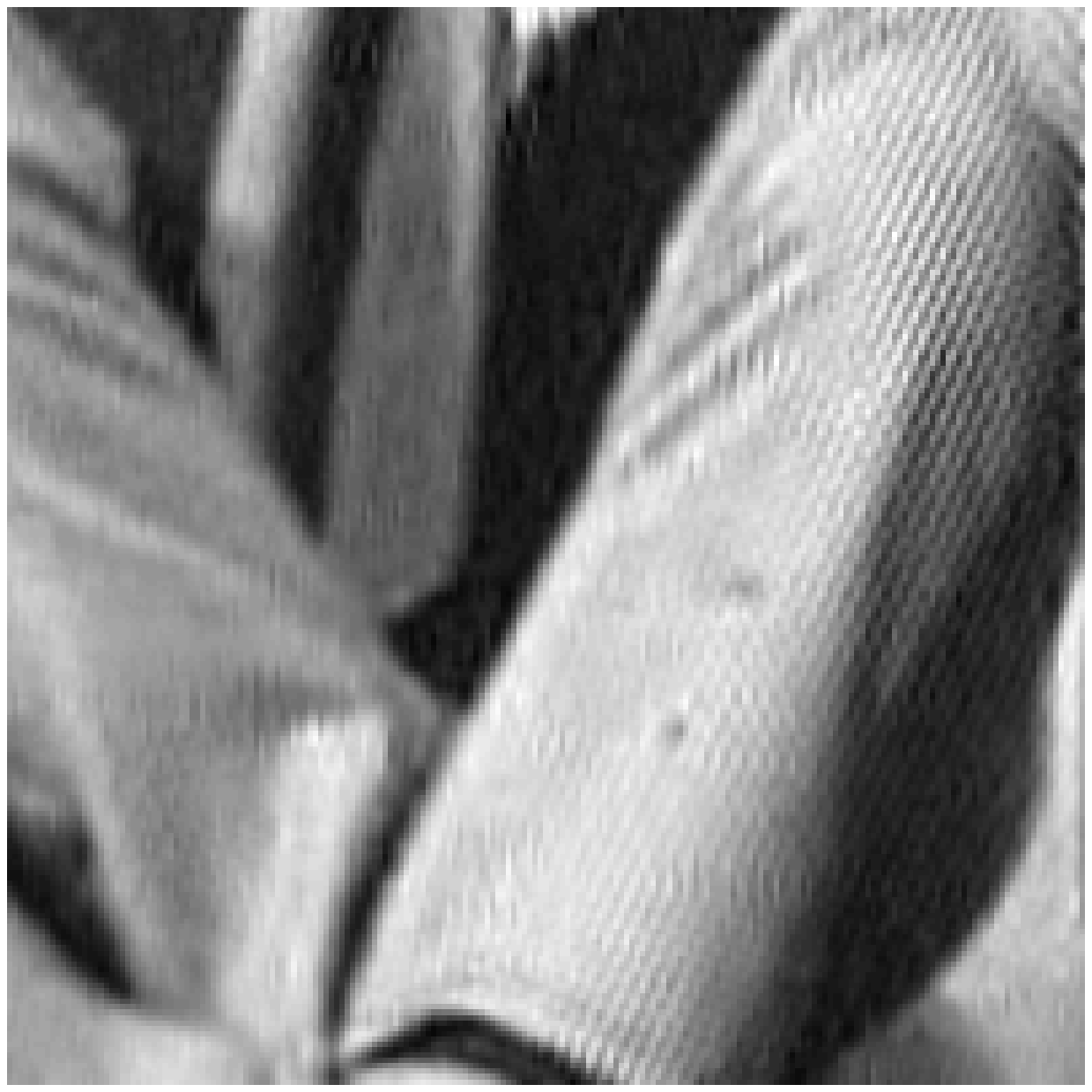} \\
(c) & (d) & (e)\\
\end{tabular}
\caption{Zooms on Barbara image, $\mathrm{BSNR}=25$ dB; (a) Original, (b) Degraded, (c) Restored with ForWaRD, (d) Restored with TwIST
and (e) Restored with the proposed method 
using \eqref{eq:sigmo}. \label{fig:rest_im}}
\end{center}
\end{figure*}

\begin{table*}[ht]
\caption{Tunis image restoration for various blurs.\label{tab:tunis_blur}}
\begin{center}
\begin{tabular}{|c||c||c|c|c|c|c||c||c||c|c|c|c|c|}
\hline
Blur & BSNR & 10 & 15 &  20 & 25 & 30 & Blur & BSNR & 10 & 15 &  20 & 25 & 30\\
\hline
\hline
 & SNR\_i & 9.662 & 13.86 & 17.01 & 18.80 & 19.56  & & SNR\_i & 9.577 & 13.64 & 16.56 & 18.13 & 18.77 \\
 & SNR\_b & \textbf{18.16} & \textbf{19.11} & \textbf{20.00} & \textbf{20.83} & \textbf{21.57} & & SNR\_b & \textbf{18.06} & \textbf{18.87} & \textbf{19.57} & \textbf{20.30} & \textbf{21.14} \\
Gaussian & SNR\_s & \textbf{18.16} & 19.10 & 19.99 & 20.81 & 21.56 &  Uniform & SNR\_s & 18.05 & 18.86 & \textbf{19.57} & 20.29 & \textbf{21.14} \\ 
$\sigma_h=2$ & SNR\_f & 17.52 & 18.57 & 19.61 & 20.57 & 21.43 & $7 \times 7$ & SNR\_f & 16.54 & 18.02 & 19.15 & 20.10 & 20.97 \\ 
 & SNR\_t & 17.12 & 18.32 & 19.47 & 20.40 & 21.33 &  & SNR\_t & 17.39 & 18.21 & 19.22 & 20.04 & 20.94 \\
 & SNR\_v & 17.39 & 18.98 & 19.84 & 20.63 & 21.45  & & SNR\_v & 17.39 & 18.81 & 19.53 & 20.27 & 21.13 \\ 
& SNR\_w & 16.24 &  18.70 & 19.49 & 19.57 & 19.62 & & SNR\_w & 16.06 &  18.45 &  19.11  & 19.28 &  19.33 \\ 
 & SNR\_r & 17.73 & 18.63 & 19.35 & 19.57 & 19.62 &  & SNR\_r & 17.64 & 18.50  & 19.09 & 19.26 &  19.33\\
\hline
\hline
 & SNR\_i & 9.971 & 14.87 & 19.57 & 23.74 & 26.83 & & SNR\_i & 10.00 & 15.00 & 20.00 & 25.00 & 30.00 \\ 
  & SNR\_b & \textbf{19.82} & \textbf{21.87} & \textbf{24.05} & \textbf{26.16} & \textbf{27.99} & & SNR\_b & \textbf{19.97} & \textbf{22.30} & \textbf{25.05} & \textbf{28.25} & \textbf{31.97} \\
Cosine & SNR\_s & \textbf{19.82} & 21.85 & 24.02 & 26.13 & 27.98 & & SNR\_s & 19.96 & 22.27 & 25.02 & 28.20 & 31.91 \\ 
$F_c=3/32$  & SNR\_f & 17.94 & 19.85 & 22.48 & 24.79 & 26.82 & Dirac  & SNR\_f & 5.701 & 11.68 & 19.10 & 25.96 & 31.41 \\ 
 & SNR\_t & 18.52 & 21.18 & 23.81 & 25.94 & 27.78 &  & SNR\_t & 19.82 & 22.03 & 24.27 & 27.18 & 31.00\\
 & SNR\_v & 18.30 & 21.17 & 23.45 & 25.61 & 27.41 & & SNR\_v & 18.10 & 21.02 & 23.44 & 26.58 & 30.18 \\ 
& SNR\_w & 12.63 & 17.50 & 21.98 & 25.61 & 27.89 & & SNR\_w & 10.42 &  15.13 & 20.04 & 25.01 & 30.00 \\
 & SNR\_r & 19.23 & 21.03 & 23.00 & 25.03 & 27.04 & & SNR\_r & 19.39 & 21.37 & 23.76 & 26.58 & 29.91 \\
 \hline
\end{tabular}
\end{center}
\end{table*}

\begin{table*}[ht]
\caption{Tunis image restoration for various blurs when the noise variance is estimated with the MAD estimator.\label{tab:tunis_blur_est}}
\begin{center}
\hspace*{-1cm}
\begin{tabular}{|c||c||c|c|c|c|c||c||c||c|c|c|c|c|}
\hline
Blur & BSNR & 10 & 15 &  20 & 25 & 30 & Blur & BSNR & 10 & 15 &  20 & 25 & 30\\
\hline
\hline
 & SNR\_i & 9.662 & 13.86 & 17.01 & 18.80 & 19.56  & & SNR\_i & 9.577 & 13.64 & 16.56 & 18.13 & 18.77 \\
 Gaussian & SNR\_b & 18.16 & 19.12 & 20.00 & 20.82 & 21.57 & Uniform & SNR\_b & 18.05 & 18.87 & 19.57 & 20.30 & 21.15 \\
$\sigma_h=2$ & SNR\_s & 18.15 & 19.11 & 19.99 & 20.81 & 21.56 & $7\times 7$ & SNR\_s & 18.05 & 18.86 & 19.57 & 20.30 & 21.15 \\
\hline
\hline
 & SNR\_i & 9.971 & 14.87 & 19.57 & 23.74 & 26.83 & & SNR\_i & 10.00 & 15.00 & 20.00 & 25.00 & 30.00 \\ 
 Cosine & SNR\_b &19.82 & 21.87 & 24.05 & 26.15 & 27.99 & Dirac & SNR\_b & 19.97 & 22.28 & 24.96 & 27.87 & 30.79 \\
$F_c=3/32$ & SNR\_s & 19.81 & 21.85 & 24.02 & 26.12 & 27.97 & & SNR\_s & 19.96 & 22.26 & 24.92 & 27.83 & 30.77 \\
\hline
\end{tabular}
\end{center}
\end{table*}

\section{Conclusions}
\label{se:conclu}
In this paper, we have addressed the problem of recovering data
degraded by a convolution and the addition of 
a white Gaussian noise. We have adopted a hybrid approach that combines frequency and multiscale analyses. By formulating the underlying deconvolution problem as a nonlinear estimation problem, we have shown that the involved  criterion to be optimized can be deduced from Stein's unbiased quadratic risk estimate. 
In this context, attention must
be paid to the variance of the risk estimate. The expression of this variance has been
derived in this paper. 

The flexibility of the proposed recovery approach must be  emphasized. Redundant or non-redundant data representations
can be employed  as well as various combinations of linear/nonlinear estimates. Based on a specific choice of the wavelet representation and
particular forms of the estimator structure, experiments have been conducted on a set of images, illustrating the good performance of the proposed approach.
In our future work, we plan to further improve this restoration method by considering more sophisticated forms of the estimator, for example, by taking into account multiscale or spatial dependencies as proposed in \cite{BLU_LUISIER_08,CHAUX_DUVAL_BENAZZA_PESQUET_08} for denoising problems.
Furthermore, it seems interesting to extend this work to the case of multicomponent data by accounting for the cross-channel correlations. 
 
\nopagebreak
\begin{appendices}
\section{Proof of Proposition \ref{p:steinf}}\label{a:steinf}
We first notice that Assumption \ref{as:psteinf2} is a sufficient
condition for the existence of the left hand-side terms of \eqref{eq:steinf1}-\eqref{eq:steinf5} since, by H{\"o}lder's inequality,
\begin{align}
\E[|\Theta_1(\rho_1)\widetilde{\eta}_1|]\le& 
\E[|\Theta_1(\rho_1)|^3]^{1/3}\E[|\widetilde{\eta}_1|^{3/2}]^{2/3}\\
\E[|\Theta_1(\rho_1)\widetilde{\eta}_1\widetilde{\eta}_2|]\le&
\E[|\Theta_1(\rho_1)|^3]^{1/3}\E[|\widetilde{\eta}_1\widetilde{\eta}_2|^{3/2}]^{2/3}\\
\E[ |\Theta_1(\rho_1)\widetilde{\eta}_1|\widetilde{\eta}_2^2] \le&
\E[|\Theta_1(\rho_1)|^3]^{1/3}\E[|\widetilde{\eta}_1|^{3/2}|\widetilde{\eta}_2|^3]^{2/3}\\
\E[|\Theta_1(\rho_1)\Theta_1(\rho_2)\widetilde{\eta}_1\widetilde{\eta}_2|]
\le& \E[|\Theta_1(\rho_1)|^3]^{1/3}\E[|\Theta_1(\rho_2)|^3]^{1/3}\nonumber\\
&\times \E[|\widetilde{\eta}_1\widetilde{\eta}_2|^3]^{1/3}.
\end{align}
We can decompose $\widetilde{\eta}_i$ with $i\in \{1,2\}$ as follows:
\begin{equation}
\widetilde{\eta}_i = a_i \eta_1
+ \check{\eta}_i \label{eq:linpredi}
\end{equation}
where $a_i$ is the mean-square prediction coefficient given by 
\begin{equation}
\sigma^2a_i = \E[\eta_1\widetilde{\eta}_i]
\label{eq:linpreda}
\end{equation}
with $\sigma^2 = \E[\eta_1^2]$
and, $\check{\eta}_i$ is the associated zero-mean prediction error
which is independent of $\upsilon_1$ and $\eta_1$.\footnote{Recall that
$(\eta_1,\widetilde{\eta}_1,\widetilde{\eta}_2)$ is zero-mean Gaussian.}
We deduce that
\begin{equation}
\E[ \Theta_1(\rho_1)\widetilde{\eta}_1]
=  a_1 \E[ \Theta_1(\rho_1)\eta_1].
\label{eq:steinf1m2}
\end{equation}
We can invoke Stein's principle to express $\E[ \Theta_1(\rho_1)\eta_1]$, provided that the assumptions in Proposition \ref{p:stein0} are satisfied.
To check these assumptions, we remark that,
for every $\tau \in \RR$, when $|\zeta|$ is large enough,
$
|\Theta_1(\tau+\zeta)| 
\exp\big(-\frac{\zeta^2}{2\sigma^2}\big)
\le |\Theta_1(\tau+\zeta)| \zeta^2 
\exp\big(-\frac{\zeta^2}{2\sigma^2}\big)
$,
which, owing to Assumption \ref{as:psteinf1}, implies that
$
\lim_{|\zeta| \to \infty} \Theta_1(\tau+\zeta)
\exp\big(-\frac{\zeta^2}{2\sigma^2}\big) = 0.
$
In addition, from Jensen's inequality and Assumption \ref{as:psteinf2},
$
\E[|\Theta_1'(\rho_1)|] \le \E[|\Theta_1'(\rho_1)|^3]^{1/3} < \infty.
$
Consequently, \eqref{eq:steinf0} combined with
\eqref{eq:linpreda} can be applied to simplify
\eqref{eq:steinf1m2}, so allowing us to obtain \eqref{eq:steinf1}.

Let us next prove \eqref{eq:steinf3}. From \eqref{eq:linpredi}, we get:
\begin{align}
\E[\Theta_1(\rho_1) \widetilde{\eta}_1 \widetilde{\eta}_2]
=& a_1 a_2 \E[\Theta_1(\rho_1)  \eta_1^2]+
a_1 \E[\Theta_1(\rho_1)\eta_1] \E[\check{\eta}_2]\nonumber\\
&+a_2 \E[\Theta_1(\rho_1)\eta_1] \E[\check{\eta}_1]
+\E[\Theta_1(\rho_1)] \E[\check{\eta}_1 \check{\eta}_2]\nonumber\\
=& a_1 a_2 \E[\Theta_1(\rho_1) \eta_1^2]+
\E[\Theta_1(\rho_1)] \E[\check{\eta}_1 \check{\eta}_2]
\end{align}
where we have used in the first equality the fact that $(\check{\eta}_1,\check{\eta}_2)$ is independent of $(\eta_1,\upsilon_1)$ and, in the second one, that
it is zero-mean. Then, by making use of the orthogonality relation:
\begin{equation}
\E[\widetilde{\eta}_1\widetilde{\eta}_2] = a_1 a_2 \sigma^2 + \E[\check{\eta}_1
\check{\eta}_2]
\label{eq:orthotildelta}
\end{equation}
we have
\begin{align}
\E[\Theta_1(\rho_1) \widetilde{\eta}_1 \widetilde{\eta}_2]
= &a_1 a_2 \big(\E[\Theta_1(\rho_1) \eta_1^2]-\sigma^2\E[\Theta_1(\rho_1)]\big)
\nonumber\\
&+ \E[\Theta_1(\rho_1)] \E[\widetilde{\eta}_1\widetilde{\eta}_2].
\label{eq:steinf2inter}
\end{align}
In addition, by integration by parts,
the conditional expectation w.r.t. $\upsilon_1$ given by
\begin{multline}
\forall \tau,\qquad
\E[\Theta_1(\rho_1) \eta_1^2 \mid \upsilon_1=\tau]
=\\ \frac{1}{\sqrt{2\pi}\sigma}\int_{-\infty}^\infty \Theta_1(\tau+\zeta) \zeta^2 \exp\big(-\frac{\zeta^2}{2\sigma^2}\big)\,d\zeta
\end{multline}
can be reexpressed as
\begin{align}
&\E[\Theta_1(\rho_1) \eta_1^2 \mid \upsilon_1=\tau]\nonumber\\
 = &\frac{\sigma}{\sqrt{2\pi}}\Big(\lim_{\zeta \to -\infty} \Theta_1(\tau+\zeta) \zeta \exp\big(-\frac{\zeta^2}{2\sigma^2}\big) \nonumber\\
& -\lim_{\zeta \to \infty} \Theta_1(\tau+\zeta) \zeta \exp\big(-\frac{\zeta^2}{2\sigma^2}\big)\nonumber\\
&+ \int_{-\infty}^\infty \big(\Theta_1(\tau+\zeta)+\Theta_1'(\tau+\zeta) \zeta\big) \exp\big(-\frac{\zeta^2}{2\sigma^2}\big)\,d\zeta\Big).
\label{eq:ippsteinf2}
\end{align}
The existence of the latter integral is secured for almost every value 
$\tau$ that can be taken by $\upsilon_1$, thanks to Assumptions \ref{as:psteinf2} and \ref{as:psteinf3} and, the fact
that, if $\mu$ denotes the probability measure of $\upsilon_1$,
\begin{align}
&\iint_{\mathbb{R}^2} 
\big|\Theta_1(\upsilon_1+\tau)+\Theta_1'(\upsilon_1+\tau) \zeta\big| \exp\big(-\frac{\zeta^2}{2\sigma^2}\big)\,d\zeta d\mu(\tau)\nonumber\\
= &\E[|\Theta_1(\rho_1)+\Theta_1'(\rho_1) \eta_1|]\nonumber\\
\le& \E[|\Theta_1(\rho_1)|]+\E[|\Theta_1'(\rho_1) \eta_1|]\nonumber\\
 \le& \E[|\Theta_1(\rho_1)|^3]^{1/3}+
\E[|\Theta_1'(\rho_1)|^3]^{1/3}\E[|\eta_1|^{3/2}]^{2/3} < \infty.
\end{align}
Since, for every $\tau \in \RR$, when $|\zeta|$ is large enough,
$
|\Theta_1(\tau+\zeta)\zeta| 
\exp\big(-\frac{\zeta^2}{2\sigma^2}\big)
\le |\Theta_1(\tau+\zeta)| \zeta^2 
\exp\big(-\frac{\zeta^2}{2\sigma^2}\big)$,
Assumption \ref{as:psteinf1} implies that
$
\lim_{|\zeta| \to \infty}\Theta_1(\tau+\zeta)\zeta \exp\big(-\frac{\zeta^2}{2\sigma^2}\big)=0$.
By using this property, we deduce from \eqref{eq:ippsteinf2}
that
$
\E[\Theta_1(\rho_1) \eta_1^2 \mid \upsilon_1]
 = \sigma^2\big(\E[\Theta_1(\upsilon_1+\eta_1)\mid \upsilon_1]+\E[\Theta_1'(\upsilon_1+\eta_1) \eta_1\mid \upsilon_1]\big)
$,
which yields
\begin{equation}
\E[\Theta_1(\rho_1) \eta_1^2]
 = \sigma^2\big(\E[\Theta_1(\rho_1)]+\E[\Theta_1'(\rho_1) \eta_1]\big).
\label{eq:steineta2}
\end{equation}
By inserting this equation in \eqref{eq:steinf2inter}, we find that
$
\E[\Theta_1(\rho_1) \widetilde{\eta}_1 \widetilde{\eta}_2]
= a_1 a_2 \sigma^2 \E[\Theta_1'(\rho_1) \eta_1]
+ \E[\Theta_1(\rho_1)] \E[\widetilde{\eta}_1\widetilde{\eta}_2].
$
Formula \eqref{eq:steinf3} straightforwardly follows by noticing
that, according to \eqref{eq:linpredi} and \eqref{eq:linpreda},
$
\E[\Theta_1'(\rho_1)\widetilde{\eta}_2]\E[\eta_1\widetilde{\eta}_1] =
a_1a_2 \sigma^2 \E[\Theta_1'(\rho_1) \eta_1].
$
Consider now \eqref{eq:steinf4}. By using \eqref{eq:linpredi}
and the independence between $(\check{\eta}_1,\check{\eta}_2)$ and $(\eta_1,\upsilon_1)$,
we can write
\begin{align}
&\E[ \Theta_1(\rho_1)\widetilde{\eta}_1\widetilde{\eta}_2^2]\nonumber\\
=& a_1 a_2^2 \E[\Theta_1(\rho_1) \eta_1^3]
+2 a_1 a_2 \E[\Theta_1(\rho_1) \eta_1^2]\E[\check{\eta}_2]\nonumber\\
&+ a_1 \E[\Theta_1(\rho_1)\eta_1]\E[\check{\eta}_2^2]
+ a_2^2 \E[\Theta_1(\rho_1)\eta_1^2]\E[\check{\eta}_1]\nonumber\\
&+ 2 a_2 \E[\Theta_1(\rho_1)\eta_1]\E[\check{\eta}_1\check{\eta}_2]
+ \E[\Theta_1(\rho_1)]\E[\check{\eta}_1\check{\eta}_2^2]\nonumber\\
=& a_1 a_2^2 \E[\Theta_1(\rho_1) \eta_1^3]
+ \E[\Theta_1(\rho_1)\eta_1]\big(a_1\E[\check{\eta}_2^2]
+ 2 a_2 \E[\check{\eta}_1\check{\eta}_2]\big)
\label{eq:rhoth1th2}
\end{align}
where the latter equality stems from the symmetry of the probability 
distribution of $(\check{\eta}_1,\check{\eta}_2)$.
Taking into account the relation
$
\E[\widetilde{\eta}_2^2]=a_2^2\sigma^2+\E[\check{\eta}_2^2]
$
and \eqref{eq:orthotildelta}, we get
\begin{align}
\E[ \Theta_1(\rho_1)\widetilde{\eta}_1\widetilde{\eta}_2^2]
= &a_1 a_2^2 \big(\E[\Theta_1(\rho_1) \eta_1^3]-3\sigma^2\E[\Theta_1(\rho_1)\eta_1]\big)\nonumber\\
&+\E[\Theta_1(\rho_1)\eta_1]\big(a_1\E[\widetilde{\eta}_2^2]+
2 a_2 \E[\widetilde{\eta}_1\widetilde{\eta}_2]\big).
\end{align}
Let us now focus our attention on $\E[\Theta_1(\rho_1) \eta_1^3]$.
Since Assumptions \ref{as:psteinf2} and \ref{as:psteinf3} imply that
\begin{multline}
\E[|2\Theta_1(\rho_1)\eta_1+\Theta_1'(\rho_1) \eta_1^2|]
\le 2\E[|\Theta_1(\rho_1)|^3]^{1/3}\E[|\eta_1|^{3/2}]^{2/3}\\ + 
\E[|\Theta_1'(\rho_1)|^3]^{1/3}\E[|\eta_1|^3]^{2/3} < \infty
\end{multline}
and Assumption \ref{as:psteinf1} holds, we can proceed by integration by parts,
similarly to the proof of \eqref{eq:steineta2} to show that
\begin{equation}
\E[\Theta_1(\rho_1) \eta_1^3] = \sigma^2\big(2\E[\Theta_1(\rho_1)\eta_1]+
\E[\Theta_1'(\rho_1) \eta_1^2]\big).
\end{equation}
Thus, \eqref{eq:rhoth1th2} reads
\begin{align}
\E[ \Theta_1(\rho_1)\widetilde{\eta}_1\widetilde{\eta}_2^2]
= &a_1 a_2^2\sigma^2\E[\Theta_1'(\rho_1) \eta_1^2]
+\E[\Theta_1(\rho_1)\eta_1]\nonumber\\
& \times \big(a_1\E[\widetilde{\eta}_2^2]
-a_1 a_2^2\sigma^2
+2 a_2 \E[\widetilde{\eta}_1\widetilde{\eta}_2] \big)
\end{align}
which, by using \eqref{eq:steinf0}, can also be reexpressed as
\begin{align}
\E[ \Theta_1(\rho_1)\widetilde{\eta}_1\widetilde{\eta}_2^2]
= &a_1 a_2^2\sigma^2\E[\Theta_1'(\rho_1) \eta_1^2]
+\sigma^2\E[\Theta_1'(\rho_1)]\nonumber\\
&\times\big(a_1\E[\widetilde{\eta}_2^2]
-a_1 a_2^2\sigma^2+
2 a_2 \E[\widetilde{\eta}_1\widetilde{\eta}_2] \big).
\label{eq:steinf41}
\end{align}
In turn, we have
\begin{align}
\E[ \Theta_1'(\rho_1)\widetilde{\eta}_2^2]
=& a_2^2 \E[ \Theta_1'(\rho_1)\eta_1^2]+2a_2\E[\Theta_1'(\rho_1)\eta_1]\E[\check{\eta}_2]\nonumber\\
&+ \E[ \Theta_1'(\rho_1)]\E[\check{\eta}_2^2]\nonumber\\
=& a_2^2 \E[ \Theta_1'(\rho_1)\eta_1^2]+ 
\E[ \Theta_1'(\rho_1)]\big(\E[\widetilde{\eta}_2^2]-a_2^2 \sigma^2\big)
\end{align}
which, by using \eqref{eq:linpreda}, leads to
\begin{align}
\E[ \Theta_1'(\rho_1)\widetilde{\eta}_2^2]\E[\eta_1\widetilde{\eta}_1]
=& a_1a_2^2\sigma^2 \E[\Theta_1'(\rho_1)\eta_1^2]\nonumber\\
&+ \sigma^2\E[ \Theta_1'(\rho_1)]a_1\big(\E[\widetilde{\eta}_2^2]-a_2^2 \sigma^2\big).
\label{eq:steinf42}
\end{align}
From the difference of \eqref{eq:steinf41} and \eqref{eq:steinf42}, we derive that
\begin{align}
\E[ \Theta_1(\rho_1)\widetilde{\eta}_1\widetilde{\eta}_2^2]
= &\E[ \Theta_1'(\rho_1)\widetilde{\eta}_2^2]\E[\eta_1\widetilde{\eta}_1]\nonumber\\
&+2a_2 \sigma^2\E[\Theta_1'(\rho_1)]\E[\widetilde{\eta}_1\widetilde{\eta}_2]
\end{align}
which, by using again \eqref{eq:linpreda}, yields \eqref{eq:steinf4}.

Finally, we will prove Formula \eqref{eq:steinf5}. 
We decompose $\widetilde{\eta}_1$ as follows:
\begin{equation}
\widetilde{\eta}_1 = b\,\widetilde{\eta}_2 + \widetilde{\eta}_1^\perp
\qquad\text{where}\qquad
b\,\widetilde{\sigma}^2 = \E[\widetilde{\eta}_1\widetilde{\eta}_2],
\label{eq:te1p}
\end{equation}
$\widetilde{\sigma}^2 = \E[ \widetilde{\eta}_2^2]$
and, $\widetilde{\eta_1}^\perp$ is independent of $(\widetilde{\eta}_2,\upsilon_1,\upsilon_2)$.
This allows us to write
\begin{multline}
\E[\Theta_1(\rho_1)\Theta_2(\rho_2)\widetilde{\eta}_1\widetilde{\eta}_2] =
b\,\E[\Theta_1(\rho_1)\Theta_2(\rho_2)\widetilde{\eta}_2^2]\\
+ \E[\Theta_1(\rho_1)\Theta_2(\rho_2)\widetilde{\eta}_1^\perp\widetilde{\eta}_2].
\label{eq:t1t2tet1tet2}
\end{multline}
Let us first calculate $\E[\Theta_1(\rho_1)\Theta(\rho_2)\widetilde{\eta}_2^2]$.
For $i\in \{1,2\}$, consider the decomposition:
\begin{equation}
\eta_i = c_i \widetilde{\eta}_2 + \eta_i^\perp
\label{eq:defetip}
\end{equation}
where $c_i \widetilde{\sigma}^2 = \E[\eta_i\widetilde{\eta}_2]$
and, $\widetilde{\eta}_2$,
$(\eta_1^\perp,\eta_2^\perp)$ and $(\upsilon_1,\upsilon_2)$ are independent.
We have then
\begin{multline}
\E[\Theta_1(\rho_1)\Theta_2(\rho_2)\widetilde{\eta}_2^2\mid \eta_1^\perp,\eta_2^\perp,\upsilon_1,\upsilon_2]\\
= \frac{1}{\sqrt{2\pi}\widetilde{\sigma}}
\int_{-\infty}^\infty \Theta_1(\upsilon_1+c_1 \zeta + \eta_1^\perp)
\Theta_2(\upsilon_2+c_2 \zeta + \eta_2^\perp)\\
\times \zeta^2\exp\big(-\frac{\zeta^2}{2\widetilde{\sigma}^2}\big)\,d\zeta.
\label{eq:expcondmult}
\end{multline}
It can be noticed that
\begin{align}
&\E[|\Theta_1(\rho_1) \Theta_2(\rho_2)+(c_1\Theta_1'(\rho_1) \Theta_2(\rho_2)
+c_2\Theta_1(\rho_1) \Theta_2'(\rho_2))\,\widetilde{\eta}_2|]
\nonumber\\
&\le \E[|\Theta_1(\rho_1)|^3]^{1/3}\E[|\Theta_2(\rho_2)|^3]^{1/3}
+|c_1|\big(\E[|\Theta_1'(\rho_1)|^3]^{1/3}\nonumber\\
&\times \E[|\Theta_2(\rho_2)|^3]^{1/3}
+|c_2|\E[|\Theta_1(\rho_1)|^3]^{1/3}\E[|\Theta_2'(\rho_2))|^3]^{1/3}\big)\nonumber\\
&\times \E[|\widetilde{\eta}_2|^3]^{1/3}< \infty
\end{align}
and, for every $(\tau_1,\tau_2) \in \RR^2$,
$
\lim_{|\zeta| \to \infty}\Theta_1(\tau_1+c_1 \zeta)
\Theta_2(\tau_2+c_2 \zeta)\zeta\exp\big(-\frac{\zeta^2}{2\widetilde{\sigma}^2}\big)= 0
$
since, for $|\zeta|$ large enough,
\begin{multline}
c_1^2 c_2^2|\Theta_1(\tau_1+c_1 \zeta)
\Theta_2(\tau_2+c_2 \zeta)\zeta|\exp\big(-\frac{\zeta^2}{2\widetilde{\sigma}^2}
\big)\\
\le |\Theta_1(\tau_1+c_1 \zeta)|(c_1\zeta)^2\exp\big(-\frac{\zeta^2}{4\widetilde{\sigma}^2}\big)\\
\times|\Theta_2(\tau_2+c_2 \zeta)|(c_2\zeta)^2\exp\big(-\frac{\zeta^2}{4\widetilde{\sigma}^2}\big)
\end{multline}
and Assumption \ref{as:psteinf1} holds.
We can therefore deduce, by integrating by parts in \eqref{eq:expcondmult}
and taking the expectation w.r.t. $(\eta_1^\perp,\eta_2^\perp,\upsilon_1,\upsilon_2)$, that
\begin{align}
&\E[\Theta_1(\rho_1)\Theta_2(\rho_2)\widetilde{\eta}_2^2]\nonumber\\
=& \widetilde{\sigma}^2\big(\E[\Theta_1(\rho_1) \Theta_2(\rho_2)]+
c_1\E[\Theta_1'(\rho_1) \Theta_2(\rho_2)\,\widetilde{\eta}_2]\nonumber\\
&+c_2\E[\Theta_1(\rho_1) \Theta_2'(\rho_2)\,\widetilde{\eta}_2]\big)\nonumber\\
=& \widetilde{\sigma}^2\E[\Theta_1(\rho_1) \Theta_2(\rho_2)]+
\E[\Theta_1'(\rho_1) \Theta_2(\rho_2)\,\widetilde{\eta}_2]\E[\eta_1 \widetilde{\eta}_2]\nonumber\\
&+\E[\Theta_1(\rho_1) \Theta_2'(\rho_2)\,\widetilde{\eta}_2]\E[\eta_2 \widetilde{\eta}_2].
\label{eq:t1t2ett2}
\end{align}
Let us now calculate 
$\E[\Theta_1(\rho_1)\Theta_2(\rho_2)\widetilde{\eta}_1^\perp\widetilde{\eta}_2]$. We have, for $i\in \{1,2\}$,
$
\eta_i = \check{c}_i \widetilde{\eta}_1^\perp + \check{\eta}_i^\perp
$,
where
\begin{equation}
\check{c_i} \E[(\widetilde{\eta}_1^\perp)^2] = \E[\eta_i\widetilde{\eta}_1^\perp]
\label{eq:defcic}
\end{equation}
and, $\widetilde{\eta}_1^\perp$ is independent of $(\widetilde{\eta}_2,\check{\eta}_1^\perp,
\check{\eta}_2^\perp,\upsilon_1,\upsilon_2)$.
By proceeding similarly to the proof of \eqref{eq:t1t2ett2},
we get
\begin{multline}
\E[\Theta_1(\rho_1)\Theta_2(\rho_2)\widetilde{\eta}_1^\perp
\mid \widetilde{\eta}_2,\check{\eta}_1^\perp,
\check{\eta}_2^\perp]=
\E[(\widetilde{\eta}_1^\perp)^2]\\\times
\big(\check{c}_1\E[\Theta_1'(\rho_1)\Theta_2(\rho_2)\mid \widetilde{\eta}_2,\check{\eta}_1^\perp,\check{\eta}_2^\perp]\\+
\check{c}_2 \E[\Theta_1(\rho_1)\Theta'_2(\rho_2)\mid \widetilde{\eta}_2,\check{\eta}_1^\perp,\check{\eta}_2^\perp]\big)
\end{multline}
which, owing to \eqref{eq:defcic}, allows us to write
\begin{align}
\E[\Theta_1(\rho_1)\Theta_2(\rho_2)\widetilde{\eta}_1^\perp\widetilde{\eta}_2]
= &\E[\Theta_1'(\rho_1)\Theta_2(\rho_2)\widetilde{\eta}_2]\E[\eta_1\widetilde{\eta}_1^\perp]\nonumber\\
&+\E[\Theta_1(\rho_1)\Theta'_2(\rho_2)\widetilde{\eta}_2]
\E[\eta_2\widetilde{\eta}_1^\perp].
\end{align}
On the other hand, from \eqref{eq:te1p}, we deduce
that, for $i\in \{1,2\}$,
$
\E[\eta_i\widetilde{\eta}_1^\perp]
= \E[\eta_i\widetilde{\eta}_1] -
b\E[\eta_i\widetilde{\eta}_2]
$,
so yielding
\begin{align}
\E[\Theta_1(\rho_1)\Theta_2(\rho_2)\widetilde{\eta}_1^\perp\widetilde{\eta}_2]
= &\,\E[\Theta_1'(\rho_1)\Theta_2(\rho_2)\widetilde{\eta}_2]\E[\eta_1\widetilde{\eta}_1]\nonumber\\
&+\E[\Theta_1(\rho_1)\Theta'_2(\rho_2)\widetilde{\eta}_2]
\E[\eta_2\widetilde{\eta}_1]\nonumber\\
&- b\big( \E[\Theta_1'(\rho_1)\Theta_2(\rho_2)\widetilde{\eta}_2]\E[\eta_1\widetilde{\eta}_2]\nonumber\\
&+\E[\Theta_1(\rho_1)\Theta'_2(\rho_2)\widetilde{\eta}_2]
\E[\eta_2\widetilde{\eta}_2]\big) .
\label{eq:t1t2tet1ptet2}
\end{align}
Altogether, \eqref{eq:t1t2tet1tet2}, \eqref{eq:te1p},
\eqref{eq:t1t2ett2} and \eqref{eq:t1t2tet1ptet2}
lead to 
\begin{multline}
\E[\Theta_1(\rho_1)\Theta_2(\rho_2)\widetilde{\eta}_1\widetilde{\eta}_2]
=\E[\Theta_1(\rho_1)\Theta_2(\rho_2)]\E[\widetilde{\eta}_1\widetilde{\eta}_2]\\
+\E[\Theta_1'(\rho_1)\Theta_2(\rho_2)\widetilde{\eta}_2]\E[\eta_1\widetilde{\eta}_1]+
\E[\Theta_1(\rho_1)\Theta'_2(\rho_2)\widetilde{\eta}_2]
\E[\eta_2\widetilde{\eta}_1].
\label{eq:steinf5m2}
\end{multline}
In order to obtain a more symmetric expression,
let us now look at the difference
$\E[\Theta_1(\rho_1)\Theta'_2(\rho_2)\widetilde{\eta}_2]
\E[\eta_2\widetilde{\eta}_1]
- \E[\Theta_1(\rho_1)\Theta'_2(\rho_2)\widetilde{\eta}_1]
\E[\eta_2\widetilde{\eta}_2]=\E[\Theta_1(\rho_1)\Theta'_2(\rho_2)\widetilde{\eta}_{12}]$
where
$
\widetilde{\eta}_{12}=\E[\eta_2\widetilde{\eta}_1]\widetilde{\eta}_2-
\E[\eta_2\widetilde{\eta}_2]\widetilde{\eta}_1.
$
Since $\widetilde{\eta}_{12}$ is a linear combination
of $\widetilde{\eta}_1$ and $\widetilde{\eta}_2$ and,
$
\E[\eta_2\widetilde{\eta}_{12}] = 0,
$
$\widetilde{\eta}_{12}$ is independent of $(\eta_2,\upsilon_1,\upsilon_2)$.
Similarly to the derivation of Formula \eqref{eq:steinf1}, it can be deduced that
\begin{align}
&\E[\Theta_1(\rho_1)\Theta'_2(\rho_2)\widetilde{\eta}_{12}\mid \eta_2,\upsilon_1,\upsilon_2]\nonumber\\
=&\E[\Theta_1(\rho_1)\Theta'_2(\upsilon_2+\eta_2)\widetilde{\eta}_{12}\mid \eta_2,\upsilon_1,\upsilon_2]\nonumber\\
=& \E[\Theta_1(\rho_1)\widetilde{\eta}_{12}\mid \eta_2,\upsilon_1,\upsilon_2]\,\Theta'_2(\upsilon_2+\eta_2)\nonumber\\
=&  \E[\Theta_1'(\rho_1)\mid \eta_2,\upsilon_1,\upsilon_2]
\E[\eta_1\widetilde{\eta}_{12}]\,\Theta'_2(\rho_2)
\end{align}
which leads to
\begin{align}
&\E[\Theta_1(\rho_1)\Theta'_2(\rho_2)\widetilde{\eta}_2]
\E[\eta_2\widetilde{\eta}_1]
- \E[\Theta_1(\rho_1)\Theta'_2(\rho_2)\widetilde{\eta}_1]
\E[\eta_2\widetilde{\eta}_2]\nonumber\\
=&\E[\Theta_1(\rho_1)\Theta'_2(\rho_2)\widetilde{\eta}_{12}]
=  \E[\Theta_1'(\rho_1)\Theta'_2(\rho_2)]
\E[\eta_1\widetilde{\eta}_{12}]\nonumber\\
=& \E[\Theta_1'(\rho_1)\Theta'_2(\rho_2)] (\E[\eta_1\widetilde{\eta}_2]
\E[\eta_2\widetilde{\eta}_1]-\E[\eta_1\widetilde{\eta}_1]
\E[\eta_2\widetilde{\eta}_2]).
\label{eq:presteinf5}
\end{align}
Eq. \eqref{eq:steinf5} is then derived by combining \eqref{eq:steinf5m2} with \eqref{eq:presteinf5}.

\section{Proof of Proposition \ref{p:steininv}}\label{a:steininv}
We have, for every $\mathbf{p}\in\mathbb{D}$,
\begin{equation}
\E[|\widetilde{R}(\mathbf{p})|^2]
= \E[|S(\mathbf{p})|^2]
+\E[|\widetilde{N}(\mathbf{p})|^2]
+2\,\mathrm{Re}\{\E[\widetilde{N}(\mathbf{p})\big(S(\mathbf{p})\big)^*]\}.
\end{equation}
Since $\widetilde{n}$ and $s$ are uncorrelated, this yields
$
\E[|S(\mathbf{p})|^2]
=\E[|\widetilde{R}(\mathbf{p})|^2]- \E[|\widetilde{N}(\mathbf{p})|^2].
$
In addition, we have 
$
\E[S(\mathbf{p}) \big(\widehat{S}(\mathbf{p})\big)^*]=   
\E[\widetilde{R}(\mathbf{p}) \big(\widehat{S}(\mathbf{p})\big)^*] 
- \E[\widetilde{N}(\mathbf{p})\big(\widehat{S}(\mathbf{p})\big)^*].
$
The previous two equations show that
\begin{align}
\forall \mathbf{p} \in \mathbb{D},\qquad
&\E[|\widehat{S}(\mathbf{p})-S(\mathbf{p})|^2]
= \E[|\widehat{S}(\mathbf{p})-\widetilde{R}(\mathbf{p})|^2]\nonumber\\
-&\E[|\widetilde{N}(\mathbf{p})|^2]
+  2\,\mathrm{Re}\{\E[\widetilde{N}(\mathbf{p})\big(\widehat{S}(\mathbf{p})\big)^*]\}.
\label{eq:defDelta1}
\end{align}
Moreover, using \eqref{eq:noisestatF}, the second term in the right-hand side of \eqref{eq:defDelta1} is
\begin{align}
\E[|\widetilde{N}(\mathbf{p})|^2]
= \frac{\gamma D}{|H(\mathbf{p})|^2}.
\label{eq:firsttermsteincor} 
\end{align}
On the other hand, according to \eqref{eq:estsynthf}, the last term in the right-hand side of \eqref{eq:defDelta1} is such that
\begin{align}
\E[\widetilde{N}(\mathbf{p})\big(\widehat{S}(\mathbf{p})\big)^*]
= &\sum_{\mathbf{x}\in\mathbb{D}}\sum_{\ell=1}^L
\E[\widehat{s}_{\ell}\,\widetilde{n}(\mathbf{x})] \nonumber\\
&\times\exp(-2\pi \imath{\mathbf{x}}^\top{\boldsymbol D}^{-1}{\mathbf{p}}) 
\big(\widetilde{\Phi}_\ell({\mathbf{p}})\big)^*.
\label{eq:secondtermsteincor}
\end{align}
Furthermore, we know from \eqref{eq:NLestpart} that 
$\widehat{s}_{\ell} = \Theta_\ell(u_\ell +  n_\ell)$,
where
\begin{equation}
n_\ell = \langle n,\varphi_\ell \rangle,\qquad
u_\ell = \langle u,\varphi_\ell \rangle\label{eq:noiseanal}
\end{equation}
and, $u$ is the field in
 $\mathbb{R}^{D_1 \times \cdots \times D_d}$ whose discrete Fourier coefficients are given by \eqref{eq:convperfreq}.
From \eqref{eq:noiseanal} as well as the assumptions made on the noise $n$ corrupting the data, it is clear
that $\big(n_\ell, \widetilde{n}(\mathbf{x})\big)$ is a zero-mean Gaussian vector
which is independent of $u_\ell$.
Thus, by using \eqref{eq:steinf1} in Proposition \ref{p:steinf},
we obtain:
\begin{equation}
\E[ \widehat{s}_{\ell}\widetilde{n}(\mathbf{x})]
= \E[ \Theta'_\ell(r_\ell)] \E[n_\ell\,\widetilde{n}(\mathbf{x})].
\label{eq:steinuse1}
\end{equation}
Let us now calculate $\E[n_\ell\,\widetilde{n}(\mathbf{x})]$.
Using \eqref{eq:noiseanal} and \eqref{eq:noisestatFI}, we get
\begin{align}
&\E[n_\ell\,\widetilde{n}(\mathbf{x})]
= \sum_{\mathbf{y}\in \mathbb{D}}  
\E[\widetilde{n}(\mathbf{x})\,n(\mathbf{y})]\varphi_\ell(\mathbf{y})\nonumber\\
= &\sum_{(\mathbf{p}',\mathbf{p}'')\in \mathbb{D}^2} 
\E\big[\widetilde{N}(\mathbf{p}') \big(N(\mathbf{p}'')\big)^*\big]
\exp(2\pi \imath{\mathbf{x}}^\top{\boldsymbol D}^{-1}{\mathbf{p}'}) 
\frac{\Phi_\ell(\mathbf{p}'')}{D^2}\nonumber\\
= & \frac{\gamma}{D}\sum_{\mathbf{p}'\in \mathbb{D}} 
\frac{\Phi_\ell(\mathbf{p}')}{H(\mathbf{p}')}
\exp(2\pi \imath{\mathbf{x}}^\top{\boldsymbol D}^{-1}{\mathbf{p}'}) 
.
\label{eq:devinterterm2}
\end{align}
Combining this equation with \eqref{eq:secondtermsteincor} 
and \eqref{eq:steinuse1} yields
\begin{equation}
\E[\widetilde{N}(\mathbf{p})\big(\widehat{S}(\mathbf{p})\big)^*]
= \gamma \sum_{\ell=1}^L\E[ \Theta'_\ell(r_\ell)] 
\frac{\Phi_\ell({\mathbf{p}})
\big(\widetilde{\Phi}_\ell({\mathbf{p}})\big)^*}
{H(\mathbf{p})}.
\label{eq:secondtermsteincorsimp}
\end{equation}
Gathering now \eqref{eq:defDelta1}, \eqref{eq:firsttermsteincor}
and \eqref{eq:secondtermsteincorsimp}, 
\eqref{eq:steindecriskDp} is obtained.

From Parseval's formula,
the global MSE can be expressed as
\begin{equation}
D\,\E[\mathcal{E}(\widehat{s}-s)]  = 
\frac{1}{D}
\sum_{\mathbf{p}\in \mathbb{D}} \E[|S(\mathbf{p})-\widehat{S}(\mathbf{p})|^2].
\end{equation}
The above equation together with \eqref{eq:steindecriskDp} 
show that \eqref{eq:steindecrisk} holds with
\begin{equation}
D \Delta =  \gamma \Big(2 \sum_{\ell=1}^L \E[\Theta_\ell'(r_\ell)] 
\,\mathrm{Re}\{\overline{\gamma}_\ell\}
-\sum_{\mathbf{p}\in \mathbb{D}} | H(\mathbf{p}) |^{-2}\Big).
\end{equation}
Furthermore, by defining
\begin{equation}
\forall \ell \in \{1,\ldots,L\},\qquad
\widetilde{n}_\ell = \langle \widetilde{n},\widetilde{\varphi}_\ell\rangle
\label{eq:noisesynth}
\end{equation}
and using \eqref{eq:noisestatFI}, it can be noticed that
\begin{align}
&\E[n_\ell\,\widetilde{n}_\ell]
= \sum_{(\mathbf{x},\mathbf{y})\in \mathbb{D}^2}  
\E[\widetilde{n}(\mathbf{y}) n(\mathbf{x}) ]
\widetilde{\varphi}_\ell(\mathbf{y}) \varphi_\ell(\mathbf{x})\nonumber\\
= &\frac{1}{D^2}\sum_{(\mathbf{p},\mathbf{p'})\in \mathbb{D}^2} 
\E\big[\widetilde{N}(\mathbf{p}) \big(N(\mathbf{p}')\big)^*\big]
\big(\widetilde{\Phi}_\ell(\mathbf{p})\big)^* 
\Phi_\ell(\mathbf{p'})\nonumber\\
=&  \frac{\gamma}{D}\sum_{\mathbf{p}\in \mathbb{D}}\frac{\Phi_\ell({\mathbf{p}})
\big(\widetilde{\Phi}_\ell({\mathbf{p}})\big)^*}{H(\mathbf{p})}
= \gamma \overline{\gamma}_\ell.
\label{eq:devinterterm}
\end{align}
Hence, as claimed in the last part of our statements, $(\overline{\gamma}_\ell)_{1 \le \ell \le L}$ is real-valued since
it is the cross-correlation of a real-valued sequence.

\section{Proof of Proposition \ref{p:varstein}}\label{a:varstein}
By using \eqref{eq:estsu}, we have
$
 \sum_{\mathbf{x} \in \mathbb{D}} \widehat{s}(\mathbf{x})\,\underline{\widetilde{n}}(\mathbf{x}) 
= \sum_{\ell=1}^L \widehat{s}_\ell\,\widetilde{n}_\ell
$, where $\widetilde{n}_\ell$ has been here redefined as
\begin{equation}
\widetilde{n}_\ell = \langle \widetilde{n},\underline{\widetilde{\varphi}}_\ell \rangle.
\label{eq:defnlb}
\end{equation}
In addition,
$
\mathcal{E}(\underline{\widetilde{n}}) = 
D^{-2} \sum_{\mathbf{p}\in \mathbb{Q}} |N(\mathbf{p})|^2/|H(\mathbf{p})|^2.
$
This allows us to rewrite \eqref{eq:difE0} as
$
\mathcal{E}_o-\widehat{\mathcal{E}}_o
= -2A-B+2C
$, where
\begin{align}
A = &\frac{1}{D}\sum_{\mathbf{x} \in \mathbb{D}} s(\mathbf{x})\,\underline{\widetilde{n}}(\mathbf{x})\\
B =&  \frac{1}{D} \sum_{\mathbf{x}\in \mathbb{D}} 
\big(\underline{\widetilde{n}}(\mathbf{x})\big)^2 - \frac{\gamma}{D} \sum_{\mathbf{p}\in \mathbb{Q}}|H(\mathbf{p})|^{-2}\nonumber\\
=&\frac{1}{D} \sum_{\mathbf{p}\in \mathbb{Q}}
\frac{|N(\mathbf{p})|^2/D-\gamma}{|H(\mathbf{p})|^2}\\
C =& \frac{1}{D} \sum_{\ell=1}^L \big(\widehat{s}_\ell\,\widetilde{n}_\ell
-\Theta_\ell'(r_\ell) \gamma \overline{\gamma}_\ell\big).
\end{align}
The variance of the error in the estimation of the risk is
thus given by
\begin{multline}
\mathsf{Var}[\mathcal{E}_o-\widehat{\mathcal{E}}_o]
= \E[(\mathcal{E}_o-\widehat{\mathcal{E}}_o)^2]
= 4\E[A^2]+4\E[AB]-8\E[AC]\\
+\E[B^2]-4\E[BC]+4\E[C^2].
\label{eq:varsteini}
\end{multline}
We will now calculate each of the terms in the right hand-side term of
the above expression to determine the variance.
\begin{itemize}
\item Due to the independence of $s$ and $n$,
the first term to calculate is equal to
\begin{equation}
\E[A^2] = \frac{1}{D^4} \sum_{(\mathbf{p},\mathbf{p}') \in \mathbb{Q}^2}
\E[S(\mathbf{p})\big(S(\mathbf{p}')\big)^*]\,
\E[\widetilde{N}(\mathbf{p})\big(\widetilde{N}(\mathbf{p}')\big)^*].
\end{equation}
By using \eqref{eq:noisestatF}, this expression simplifies as
\begin{equation}
\E[A^2] = \frac{\gamma}{D^3} \sum_{\mathbf{p}\in \mathbb{Q}}
\frac{\E[|S(\mathbf{p})|^2]}{|H(\mathbf{p})|^2}.
\label{eq:EA2}
\end{equation}
\item The second term cancels. Indeed, since $n$ and hence $\underline{\widetilde{n}}$ are
zero-mean,
\begin{equation}
\E[AB] = \frac{1}{D^2} \sum_{(\mathbf{x},\mathbf{x}')\in \mathbb{D}^2}
\E[s(\mathbf{x})] \E[\underline{\widetilde{n}}(\mathbf{x})
\big(\underline{\widetilde{n}}(\mathbf{x}'))^2].
\label{eq:EAB}
\end{equation}
and, since 
$(\underline{\widetilde{n}}(\mathbf{x}),\underline{\widetilde{n}}(\mathbf{x}'))$ is zero-mean Gaussian, it has a symmetric distribution and
$\E[\underline{\widetilde{n}}(\mathbf{x})
\big(\underline{\widetilde{n}}(\mathbf{x}'))^2] = 0$.
\item The calculation of the third term is a bit more involved.
We have
\begin{multline}
\E[AC] = \frac{1}{D^2}\sum_{\mathbf{x}\in\mathbb{D}}\sum_{\ell=1}^L
\big(\E[s(\mathbf{x})\widehat{s}_\ell\,\widetilde{n}_\ell\,\underline{\widetilde{n}}(\mathbf{x})]\\ - \E[s(\mathbf{x})\Theta'_\ell(r_\ell)\underline{\widetilde{n}}(\mathbf{x})]\gamma \overline{\gamma}_\ell\big).
\label{eq:EAC0}
\end{multline}
In order to find a tractable expression of $\E[s(\mathbf{x})\widehat{s}_\ell\,\widetilde{n}_\ell\,\underline{\widetilde{n}}(\mathbf{x})]$
with $\ell \in \{1,\ldots,L\}$, we will first consider
the following conditional expectation w.r.t. $s$:
$
\E[s(\mathbf{x})\widehat{s}_\ell\,\widetilde{n}_\ell\,\underline{\widetilde{n}}(\mathbf{x})\mid s] =s(\mathbf{x})\, \E[\Theta(r_\ell)\widetilde{n}_\ell\,\underline{\widetilde{n}}(\mathbf{x})\mid s].
$
According to Formula \eqref{eq:steinf3} in Proposition \ref{p:steinf},\footnote{Proposition \ref{p:steinf} is applicable to the calculation of the conditional expectation since conditioning w.r.t. $s$ amounts to fixing $u_\ell$
(see the remark at the end of Section \ref{s:steinid}).}
\begin{align}
\E[\Theta(r_\ell)\widetilde{n}_\ell\,\underline{\widetilde{n}}(\mathbf{x})\mid s]=&\E[\Theta'_\ell(r_\ell)\underline{\widetilde{n}}(\mathbf{x})\mid s]\,\E[n_\ell\,\widetilde{n}_\ell]\nonumber\\
&+\E[\Theta_\ell(r_\ell)\mid s]\,\E[\widetilde{n}_\ell\,\underline{\widetilde{n}}(\mathbf{x})]
\end{align}
which, by using \eqref{eq:devinterterm}, allows us to deduce that
\begin{multline}
\E[s(\mathbf{x})\Theta(r_\ell)\widetilde{n}_\ell\,\underline{\widetilde{n}}(\mathbf{x})]=\E[s(\mathbf{x})\Theta'_\ell(r_\ell)\underline{\widetilde{n}}(\mathbf{x})]\,\gamma \overline{\gamma}_\ell\\
+\E[s(\mathbf{x})\Theta_\ell(r_\ell)]\,\E[\widetilde{n}_\ell\,\underline{\widetilde{n}}(\mathbf{x})].
\end{multline}
This shows that \eqref{eq:EAC0} can be simplified as follows:
\begin{equation}
\E[AC] = \frac{1}{D^2}\sum_{\mathbf{x}\in\mathbb{D}}\sum_{\ell=1}^L
\E[s(\mathbf{x})\Theta_\ell(r_\ell)]\,\E[\widetilde{n}_\ell\,\underline{\widetilde{n}}(\mathbf{x})].
\end{equation}
Furthermore, according to \eqref{eq:noisestatF} and \eqref{eq:defnlb}, we have
\begin{align}
\E[\widetilde{n}_\ell\,\underline{\widetilde{n}}(\mathbf{x})]
=& \frac{1}{D^2}\sum_{(\mathbf{p},\mathbf{p}')\in \mathbb{Q}^2}\E[\widetilde{N}(\mathbf{p})
\big(\widetilde{N}(\mathbf{p}')\big)^*]\big(\widetilde{\Phi}_\ell(\mathbf{p})\big)^*\nonumber\\
&\times \exp(-2\pi \imath{\mathbf{x}}^\top{\boldsymbol D}^{-1}{\mathbf{p}}')\nonumber\\
=& \frac{\gamma}{D}\sum_{\mathbf{p}\in \mathbb{Q}}
\frac{\big(\widetilde{\Phi}_\ell(\mathbf{p})\big)^*}{|H(\mathbf{p})|^2} 
\exp(-2\pi \imath{\mathbf{x}}^\top{\boldsymbol D}^{-1}{\mathbf{p}}).
\label{eq:nltnxp}
\end{align}
This yields
\begin{align}
\E[AC] &= \frac{\gamma}{D^3}\sum_{\mathbf{p}\in\mathbb{Q}}\sum_{\ell=1}^L
\frac{\big(\widetilde{\Phi}_\ell(\mathbf{p})\big)^*}{|H(\mathbf{p})|^2} \E[S(\mathbf{p})\Theta_\ell(r_\ell)]\nonumber\\
&= \frac{\gamma}{D^3}\sum_{\mathbf{p}\in\mathbb{Q}}
\frac{\E[S(\mathbf{p})\big(\widehat{S}(\mathbf{p})\big)^*]}{|H(\mathbf{p})|^2}.
\label{eq:EAC}
\end{align}
\item The calculation of the fourth term is more classical  since $|N(\mathbf{p})|^2/D$ is the $\mathbf{p}$ bin of the periodogram \cite{JENK_WATT_68} of the Gaussian white noise $n$.
More precisely, since $|N(\mathbf{p})|^2/D$ is an unbiased estimate of $\gamma$,
\begin{equation}
\E[B^2] = \frac{1}{D^4} 
\sum_{(\mathbf{p},\mathbf{p}')\in \mathbb{Q}^2}\frac{\mathsf{Cov}(|N(\mathbf{p})|^2,
|N(\mathbf{p}')|^2)}{|H(\mathbf{p})|^2 |H(\mathbf{p}')|^2}.
\end{equation}
In the above summation,
we know that, if $\mathbf{p} \neq \mathbf{p}'$ and $\mathbf{p} \neq
{\boldsymbol D}\mathbf{1}-\mathbf{p}'$ with $\mathbf{1}=(1,\ldots,1)^\top \in \RR^d$,
$N(\mathbf{p})$ and $N(\mathbf{p}')$ are independent and thus,
$\mathsf{Cov}(|N(\mathbf{p})|^2,|N(\mathbf{p}')|^2) = 0$.
On the other hand,  if $\mathbf{p} = \mathbf{p}'$ or $\mathbf{p} =
{\boldsymbol D}\mathbf{1}-\mathbf{p}'$, then
$\mathsf{Cov}(|N(\mathbf{p})|^2,|N(\mathbf{p}')|^2)= \E[|N(\mathbf{p})|^4]
-\gamma^2 D^2$. Let
\begin{multline}
\mathbb{S} = \big\{\mathbf{p} = (p_1,\ldots,p_d)^\top \in \mathbb{D} \mid\\
\forall i \in \{1,\ldots,d\}, p_i \in \{0,D_i/2\}\big\}.
\end{multline} 
If $\mathbf{p} \in \mathbb{S}$, then $N(\mathbf{p})$ is a zero-mean
Gaussian real random variable and $\E[|N(\mathbf{p})|^4] = 
3 \E\big[\big(N(\mathbf{p})\big)^2\big]^2 = 3 \gamma^2 D^2$. Otherwise, $N(\mathbf{p})$
a zero-mean Gaussian circular complex random variable and
$\E[|N(\mathbf{p})|^4] = 
2 \E[|N(\mathbf{p})|^2]^2 $ $= 2 \gamma^2 D^2$.
It can be deduced that
\begin{align}
\E[B^2] =& \frac{1}{D^4} \left(
\sum_{\mathbf{p}\in \mathbb{Q} \cap \mathbb{S}} 
\frac{\mathsf{Var}[\big(N(\mathbf{p})\big)^2]}{|H(\mathbf{p})|^4}\right.\nonumber\\
&+ \sum_{\mathbf{p}\in \mathbb{Q} \cap (\mathbb{D}\setminus\mathbb{S})}
\Big(\frac{\mathsf{Var}[\big|N(\mathbf{p})\big|^2]}{|H(\mathbf{p})|^4}\nonumber\\
&\left.+\frac{\mathsf{Cov}(|N(\mathbf{p})|^2,
|N({\boldsymbol D}\mathbf{1}-\mathbf{p})|^2)}{|H(\mathbf{p})|^2 |H({\boldsymbol D}\mathbf{1}-\mathbf{p})|^2}\Big)\right)\nonumber\\
=& \frac{1}{D^4} \left(
\sum_{\mathbf{p}\in \mathbb{Q} \cap \mathbb{S}} 
\frac{2\gamma^2D^2}{|H(\mathbf{p})|^4}\right.\nonumber\\
&\left.+
\sum_{\mathbf{p}\in \mathbb{Q} \cap (\mathbb{D}\setminus\mathbb{S})}
\Big(\frac{\gamma^2D^2}{|H(\mathbf{p})|^4}+
\frac{\gamma^2D^2}{|H(\mathbf{p})|^4}\Big)\right)\nonumber\\
=& \frac{2\gamma^2}{D^2}\sum_{\mathbf{p}\in \mathbb{Q}}\frac{1}{|H(\mathbf{p})|^4}.
\label{eq:EB2}
\end{align}
\item Let us now turn our attention to the fifth term. According
to \eqref{eq:steinf1} and the definition of $\overline{\gamma}_\ell$ in
\eqref{eq:devinterterm}, for every $\ell \in \{1,\ldots,L\}$, $\widehat{s}_\ell\,\widetilde{n}_\ell
-\Theta_\ell'(r_\ell) \gamma\overline{\gamma}_\ell$ is zero-mean and we have then
\begin{multline}
\E[BC] = \frac{1}{D^2}\sum_{\mathbf{x} \in \mathbb{D}} \sum_{\ell =1}^L
\Big(\E[\widehat{s}_\ell\,\widetilde{n}_\ell \big(\underline{\widetilde{n}}(\mathbf{x})\big)^2]\\
- \E[\Theta_\ell'(r_\ell)\big(\underline{\widetilde{n}}(\mathbf{x})\big)^2]\gamma\overline{\gamma}_\ell\Big).
\label{eq:EBC1}
\end{multline}
By applying now Formula \eqref{eq:steinf4} in Proposition \ref{p:steinf}, 
we have, for every $\ell \in \{1,\ldots,L\}$,
\begin{align}
&\E[\widehat{s}_\ell\,\widetilde{n}_\ell \big(\underline{\widetilde{n}}(\mathbf{x})\big)^2]
- \E[\Theta_\ell'(r_\ell)\big(\underline{\widetilde{n}}(\mathbf{x})\big)^2]\gamma\overline{\gamma}_\ell\nonumber\\
=& \E[\Theta_\ell(r_\ell)\,\widetilde{n}_\ell \big(\underline{\widetilde{n}}(\mathbf{x})\big)^2]
- \E[\Theta_\ell'(r_\ell)\big(\widetilde{n}(\mathbf{x})\big)^2]\,
\E[n_\ell\, \widetilde{n}_\ell]\nonumber\\
=& 2 \E[\Theta_\ell'(r_\ell)]\,\E[\widetilde{n}_\ell\,\underline{\widetilde{n}}(\mathbf{x})]\,\E[\underline{\widetilde{n}}(\mathbf{x})\,n_\ell]
\label{eq:Esnn2}
\end{align}
where, in compliance with \eqref{eq:modrl}, $n_\ell$ is now given by
\begin{equation}
n_\ell = \langle \underline{n},\varphi_\ell\rangle.
\label{eq:redefnl}
\end{equation}
Furthermore, similarly to \eqref{eq:devinterterm2}, we have
\begin{align}
\E[\underline{\widetilde{n}}(\mathbf{x})\,n_\ell]
= &\frac{1}{D^2}\sum_{(\mathbf{p},\mathbf{p}')\in \mathbb{Q}^2}\E[\widetilde{N}(\mathbf{p})
\big(N(\mathbf{p}')\big)^*]\Phi_\ell(\mathbf{p}')\nonumber\\
&\times \exp(2\pi \imath{\mathbf{x}}^\top{\boldsymbol D}^{-1}{\mathbf{p}})\nonumber\\
= &\frac{\gamma}{D}\sum_{\mathbf{p}'\in \mathbb{Q}}
\frac{\Phi_\ell(\mathbf{p}')}{H(\mathbf{p}')} 
\exp(2\pi \imath{\mathbf{x}}^\top{\boldsymbol D}^{-1}{\mathbf{p}'}).
\label{eq:135manq}
\end{align}
Altogether, \eqref{eq:nltnxp}, \eqref{eq:Esnn2} 
and \eqref{eq:135manq} 
yield
\begin{multline}
\E[\widehat{s}_\ell\,\widetilde{n}_\ell \big(\underline{\widetilde{n}}(\mathbf{x})\big)^2]
- \E[\Theta_\ell'(r_\ell)\big(\underline{\widetilde{n}}(\mathbf{x})\big)^2]\gamma\overline{\gamma}_\ell\\
=\frac{2\gamma^2}{D^2} \E[\Theta_\ell'(r_\ell)]
\sum_{(\mathbf{p},\mathbf{p}')\in \mathbb{Q}^2}
\frac{\Phi_\ell(\mathbf{p}')
\big(\widetilde{\Phi}_\ell(\mathbf{p})\big)^*}{H(\mathbf{p}')|H(\mathbf{p})|^2} \\
\times\exp\big(2\pi \imath{\mathbf{x}}^\top{\boldsymbol D}^{-1}({\mathbf{p}'}-\mathbf{p})\big).
\end{multline}
Hence, \eqref{eq:EBC1} can be reexpressed as
\begin{align}
\E[BC] &=
\frac{2\gamma^2}{D^2} 
\sum_{\ell = 1}^L\E[\Theta_\ell'(r_\ell)]\, 
\kappa_\ell
\label{eq:EBC}
\end{align}
where
\begin{equation}
\kappa_\ell = \frac{1}{D}\sum_{\mathbf{p}\in\mathbb{Q}}
\frac{\Phi_\ell(\mathbf{p})
\big(\widetilde{\Phi}_\ell(\mathbf{p})\big)^*}{H(\mathbf{p})|H(\mathbf{p})|^2}.
\label{eq:defkappa}
\end{equation}
\item Let us now consider the last term
\begin{multline}
\E[C^2] = \frac{1}{D^2} \sum_{\ell=1}^L
\sum_{i=1}^L\big( \E[\widehat{s}_\ell\widehat{s}_i\widetilde{n}_\ell
\widetilde{n}_i]
-\E[\widehat{s}_i\Theta_\ell'(r_\ell)\,\widetilde{n}_i]
 \gamma\overline{\gamma}_\ell\\
-\E[\widehat{s}_\ell\Theta_i'(r_i)\,\widetilde{n}_\ell]
 \gamma\overline{\gamma}_i
+\E[\Theta_\ell'(r_\ell)\Theta_i'(r_i)]
\gamma^2\overline{\gamma}_\ell\overline{\gamma}_i\big).
\label{eq:EC2big}
\end{multline}
Appealing to Formula \eqref{eq:steinf5} in Proposition \ref{p:steinf}
and \eqref{eq:devinterterm}, we have
\begin{align}
&\E[\Theta_\ell(r_\ell)\Theta_i(r_i)\widetilde{n}_\ell\widetilde{n}_i]\nonumber\\ 
=&\,\E[\Theta_\ell(r_\ell)\Theta_i(r_i)]\E[\widetilde{n}_\ell\widetilde{n}_i]+\E[\Theta_\ell'(r_\ell)\Theta_i(r_i)\widetilde{n}_i]\gamma\overline{\gamma}_\ell\nonumber\\
&+\E[\Theta_\ell(r_\ell)\Theta'_i(r_i)\widetilde{n}_\ell]\gamma\overline{\gamma}_i+\E[\Theta_\ell'(r_\ell)\Theta'_i(r_i)]\nonumber\\
&\times (\E[n_i\widetilde{n}_\ell]
\E[\eta_\ell\widetilde{n}_i]-\gamma^2\overline{\gamma}_\ell\overline{\gamma}_i).
\end{align}
This allows us to simplify \eqref{eq:EC2big} as follows:
\begin{align}
\E[C^2] = &\frac{1}{D^2} \sum_{\ell=1}^L
\sum_{i=1}^L\big( \E[\widehat{s}_\ell\widehat{s}_i]\E[\widetilde{n}_\ell
\widetilde{n}_i]\nonumber\\
&+\E[\Theta_\ell'(r_\ell)\Theta_i'(r_i)]
\E[n_i\widetilde{n}_\ell]
\E[\eta_\ell\widetilde{n}_i]\big).
\label{eq:EC2simp}
\end{align}
Furthermore, according to \eqref{eq:noisestatF}, \eqref{eq:defnlb}, 
\eqref{eq:noisestatFI} and \eqref{eq:redefnl},
we have
\begin{align}
\E[\widetilde{n}_\ell\widetilde{n}_i]&= \frac{1}{D^2}\sum_{(\mathbf{p},\mathbf{p}')\in \mathbb{Q}^2}\E[\widetilde{N}(\mathbf{p})
\big(\widetilde{N}(\mathbf{p}')\big)^*]\big(\widetilde{\Phi}_\ell(\mathbf{p})\big)^* \widetilde{\Phi}_i(\mathbf{p}')\nonumber\\
&= \frac{\gamma}{D}\sum_{\mathbf{p}\in \mathbb{Q}}
\frac{\big(\widetilde{\Phi}_\ell(\mathbf{p})\big)^*\widetilde{\Phi}_i(\mathbf{p})}{|H(\mathbf{p})|^2}
\label{eq:nf5i}
\end{align}
and
\begin{align}
\E[n_\ell\widetilde{n}_i]
=& \frac{1}{D^2}\sum_{(\mathbf{p},\mathbf{p}')\in \mathbb{Q}^2}\E[\widetilde{N}(\mathbf{p})
\big(N(\mathbf{p}')\big)^*]\Phi_\ell(\mathbf{p}')
\big(\widetilde{\Phi}_i(\mathbf{p})\big)^*\nonumber\\
&= \gamma\overline{\gamma}_{\ell,i}
\label{eq:nf5is}
\end{align}
where the expression of $\overline{\gamma}_{\ell,i}$ is given by \eqref{eq:nf5iii}.
Hence, by using \eqref{eq:nf5i}-\eqref{eq:nf5is}, \eqref{eq:EC2simp} can be rewritten as
\begin{align}
\E[C^2] =& \frac{\gamma}{D^3} \sum_{\mathbf{p}\in \mathbb{Q}}
\frac{\E[|\widehat{S}(\mathbf{p})|^2]}{|H(\mathbf{p})|^2}\nonumber\\
&+\frac{\gamma^2}{D^2} \sum_{\ell=1}^L\
\sum_{i=1}^L\E[\Theta_\ell'(r_\ell)\Theta_i'(r_i)]
\overline{\gamma}_{\ell,i} \overline{\gamma}_{i,\ell}.
\label{eq:EC2}
\end{align}
\item In conclusion, we deduce from \eqref{eq:varsteini}, \eqref{eq:EA2},
\eqref{eq:EAC}, \eqref{eq:EB2}, \eqref{eq:EBC} and
\eqref{eq:EC2} that
\begin{align}
\mathsf{Var}[\mathcal{E}_o&-\widehat{\mathcal{E}}_o]
= \frac{4\gamma}{D^3}\sum_{\mathbf{p}\in \mathbb{Q}} 
\frac{\E[|\widehat{S}(\mathbf{p})-S(\mathbf{p})|^2]}{|H(\mathbf{p})|^2}\nonumber\\
&+\frac{4\gamma^2}{D^2}\Big( \sum_{\ell=1}^L\
\sum_{i=1}^L\E[\Theta_\ell'(r_\ell)\Theta_i'(r_i)]
\overline{\gamma}_{\ell,i} \overline{\gamma}_{i,\ell}\nonumber\\
&-2 \sum_{\ell = 1}^L\E[\Theta_\ell'(r_\ell)]\kappa_\ell
+\frac{1}{2}\sum_{\mathbf{p}\in \mathbb{Q}}\frac{1}{|H(\mathbf{p})|^4}\Big).
\label{eq:varstein} 
\end{align}
By exploiting now \eqref{eq:steindecriskDp} (see Proposition \ref{prop:ninv})
and noticing that $(\kappa_\ell)_{1 \le \ell \le L}$ is real-valued, this expression can be simplified as follows:
\begin{multline}
\mathsf{Var}[\mathcal{E}_o-\widehat{\mathcal{E}}_o]
= \frac{4\gamma}{D^3}\sum_{\mathbf{p}\in \mathbb{Q}} 
\frac{\E[|\widehat{S}(\mathbf{p})-\widetilde{R}(\mathbf{p})|^2]}{|H(\mathbf{p})|^2}\\
+\frac{4\gamma^2}{D^2} \sum_{\ell=1}^L\
\sum_{i=1}^L\E[\Theta_\ell'(r_\ell)\Theta_i'(r_i)]
\overline{\gamma}_{\ell,i} \overline{\gamma}_{i,\ell}\\
-\frac{2\gamma^2}{D^2} \sum_{\mathbf{p}\in \mathbb{Q}}\frac{1}{|H(\mathbf{p})|^4}.
\label{eq:varstein2} 
\end{multline}
Eq. \eqref{eq:varsteinf} follows by using Parseval's formula.
\end{itemize}

\end{appendices}

\end{document}

%% file: method.pstex_t
\begin{picture}(0,0)%
\includegraphics{method.pstex}%
\end{picture}%
\setlength{\unitlength}{3947sp}%
\begingroup\makeatletter\ifx\SetFigFont\undefined%
\gdef\SetFigFont#1#2#3#4#5{%
  \reset@font\fontsize{#1}{#2pt}%
  \fontfamily{#3}\fontseries{#4}\fontshape{#5}%
  \selectfont}%
\fi\endgroup%
\begin{picture}(8055,2760)(811,-3121)
\put(6226,-1456){\makebox(0,0)[b]{\smash{{\SetFigFont{11}{13.2}{\familydefault}{\mddefault}{\updefault}$(\psi_{m,\mathbf{k}_\ell})_{\ell,m}$}}}}
\put(4201,-2536){\makebox(0,0)[b]{\smash{{\SetFigFont{11}{13.2}{\familydefault}{\mddefault}{\updefault}$(\widetilde{\varphi}_\ell)_\ell$}}}}
\put(1201,-1186){\makebox(0,0)[b]{\smash{{\SetFigFont{11}{13.2}{\familydefault}{\mddefault}{\updefault}$r$}}}}
\put(3451,-1036){\makebox(0,0)[b]{\smash{{\SetFigFont{11}{13.2}{\familydefault}{\mddefault}{\updefault}$\underline{r}$}}}}
\put(6226,-2536){\makebox(0,0)[b]{\smash{{\SetFigFont{11}{13.2}{\familydefault}{\mddefault}{\updefault}$(\Theta_\ell)_\ell$}}}}
\put(1201,-2386){\makebox(0,0)[b]{\smash{{\SetFigFont{11}{13.2}{\familydefault}{\mddefault}{\updefault}$\widehat{s}$}}}}
\put(7576,-1786){\makebox(0,0)[b]{\smash{{\SetFigFont{11}{13.2}{\familydefault}{\mddefault}{\updefault}$(r_{\ell})_{\ell}$}}}}
\put(2776,-1486){\makebox(0,0)[b]{\smash{{\SetFigFont{11}{13.2}{\familydefault}{\mddefault}{\updefault}$\Pi$}}}}
\put(2776,-2536){\makebox(0,0)[b]{\smash{{\SetFigFont{11}{13.2}{\familydefault}{\mddefault}{\updefault}$\Pi$}}}}
\put(5101,-1036){\makebox(0,0)[b]{\smash{{\SetFigFont{11}{13.2}{\familydefault}{\mddefault}{\updefault}$\check{r}$}}}}
\put(4501,-1411){\makebox(0,0)[b]{\smash{{\SetFigFont{11}{13.2}{\familydefault}{\mddefault}{\updefault}$G^*$}}}}
\put(5176,-2086){\makebox(0,0)[b]{\smash{{\SetFigFont{11}{13.2}{\familydefault}{\mddefault}{\updefault}$(\widehat{s}_{\ell})_{\ell}$}}}}
\end{picture}%